\documentclass[11pt,letterpaper]{article}
\usepackage[margin=1in]{geometry}
\usepackage[utf8]{inputenc}
\usepackage[T1]{fontenc}

\usepackage{amsmath,amsfonts,amssymb,amsthm}
\usepackage{fullpage}
\usepackage{subfigure}
\usepackage{epstopdf}
\usepackage{mathtools}
\usepackage[colorlinks=true,allcolors=blue]{hyperref}
\usepackage{tcolorbox}
\usepackage{multirow}
\usepackage[shortlabels]{enumitem}
\usepackage[disable]{todonotes}
\usepackage{xspace}
\tcbuselibrary{skins,breakable}
\tcbset{enhanced jigsaw}
\usepackage[font=small,labelfont=bf]{caption}
\usepackage{microtype}
\usepackage{algorithmicx}
\usepackage{algpseudocode}
\usepackage{thmtools}

\usepackage[linesnumbered,ruled,vlined]{algorithm2e}

\usepackage{makecell}

\usepackage{thm-restate,color,xspace}
\usepackage{comment}
\usepackage{thmtools}
\usepackage{xcolor}
\usepackage{nameref}
\usepackage{array}
\usepackage{thm-restate}

\newcommand{\antonis}[1]{\todo[linecolor=orange!50!black,backgroundcolor=orange!25,bordercolor=orange!50!black]{\scriptsize \textbf{AS:} #1}}

\newcommand{\rajath}[1]{\todo[linecolor=green!50!black,backgroundcolor=green!25,bordercolor=green!50!black]{\scriptsize \textbf{RA:} #1}}
\newcommand{\yasamin}[1]{\todo[linecolor=blue!50!black,backgroundcolor=blue!25,bordercolor=blue!50!black]{\scriptsize \textbf{YA:} #1}}

\makeatletter
\newcommand{\pushright}[1]{\ifmeasuring@#1\else\omit\hfill$\displaystyle#1$\fi\ignorespaces}
\newcommand{\pushleft}[1]{\ifmeasuring@#1\else\omit$\displaystyle#1$\hfill\fi\ignorespaces}
\makeatother

\newtheorem{theorem}{Theorem}

\newtheorem{lemma}[theorem]{Lemma}
\newtheorem{observation}[theorem]{Observation}

\theoremstyle{definition}

\newtheorem{remark}[theorem]{Remark}

\newcommand{\myhs}{\hspace{0.05em}}

\renewcommand{\epsilon}{{\varepsilon}}

\newcommand{\poly}{\mathsf{poly}}

\newcommand{\dist}{\operatorname{dist}}
\newcommand{\dd}{\textnormal{\textsf{d}}}

\newif\ifhidetodos

\usepackage{marginnote}
\usepackage{booktabs}   
\usepackage{multirow}   
\usepackage{xcolor, colortbl}

\newcommand{\cbox}[2][yellow]{%
  \fcolorbox{#1}{white}{\parbox{\dimexpr\linewidth-2\fboxsep}{\strut #2\strut}}%
}

\newcommand{\Ball}[2]{B(#1,\,#2)}
\newcommand{\NBall}[2]{N(#1,\,#2)}

\DeclareMathOperator{\cost}{cost}

\newcommand{\customtodo}[3]{\textcolor{#2}{\(\blacktriangledown\)}\marginnote{\raggedright \textcolor{#2}{\textbf{#1:} #3}}}
\newcommand{\customtododisplay}[3]{\noindent\cbox[#2]{\textcolor{#2}{\textbf{#1:} #3}}}
\newcommand{\customtodoinline}[3]{\textcolor{#2}{\textcolor{#2}{\textbf{#1:} #3}}}

\ifhidetodos
    \renewcommand{\customtodo}[3]{}
    \renewcommand{\customtododisplay}[3]{}
    \renewcommand{\customtodoinline}[3]{}
\fi

\usepackage{cleveref}
\Crefname{theorem}{Theorem}{Theorems}
\Crefname{lemma}{Lemma}{Lemmas}
\Crefname{corollary}{Corollary}{Corollaries}
\Crefname{proposition}{Proposition}{Propositions}
\Crefname{observation}{Observation}{Observations}
\Crefname{claim}{Claim}{Claims}
\Crefname{fact}{Fact}{Facts}
\Crefname{assumption}{Assumption}{Assumptions}
\crefname{thm}{Theorem}{Theorems}

\crefname{thm}{Theorem}{Theorems}
\crefname{lem}{Lemma}{Lemmas}

\title{Faster Randomized and Deterministic $k$-Clustering on Graphs}

\usepackage{authblk}
\author[1]{Sebastian Forster\thanks{\texttt{forster@cs.sbg.ac.at}}}
\author[2,3]{Yasamin Nazari\thanks{\texttt{y.nazari@vu.nl}}}
\author[2]{Rajath Rao K.N.\thanks{\texttt{r.raokamblanagendra@vu.nl}}}
\author[4]{Antonis Skarlatos\thanks{\texttt{antonis.skarlatos@warwick.ac.uk}}}

\affil[1]{Department of Computer Science, University of Salzburg, Austria}
\affil[2]{Vrije Universiteit Amsterdam, The Netherlands}
\affil[3]{CWI Amsterdam, The Netherlands}
\affil[4]{Department of Computer Science, University of Warwick, Coventry, UK}

\date{\today}

\begin{document}

\maketitle
\begin{abstract} 

In this paper, we study the $(k, z)$-clustering and $k$-center problems on graphs, where $(k, z)$-clustering generalizes the $k$-median ($z = 1$) and $k$-means ($z = 2$) problems. In this work, we obtain the following main results:
\begin{itemize}
    \item Our first contribution is the first \emph{deterministic} algorithm for $k$-center on graphs that achieves a $(2 + \varepsilon)$-approximation in $\tilde{O}(m)$ time. This affirmatively resolves an open problem raised by Abboud, Cohen-Addad, Lee, and Manurangsi~\cite{abboud2023fine}. Our techniques also extend to the $k$-center with outliers problem, where up to~$t$ points may be discarded. 
    
    \item Our second contribution is a randomized algorithm for $(k,z)$-clustering on graphs that achieves an $O(1)$-approximation in $\tilde{O}(m)$ time, which in particular covers $k$-median ($z = 1$) and $k$-means ($z = 2$). Prior to this work, an $\tilde{O}(m)$-time randomized algorithm was known for $k$-median by Thorup~\cite{thorup2001quick}. Moreover, a recent work of Jiang, Jin, Lou, and Lu~\cite{jiang2026local} achieves $\Tilde{O}(m^{1+o(1)})$ time for general $z$ via local search.
    
    \item Finally, we design a \emph{deterministic} algorithm for $(k,z)$-clustering on graphs that achieves
    a $O(\poly(c))$-approximation in $\tilde{O}(m^{1+1/c})$ time, for a positive parameter $c$. To obtain this result, we use techniques from the Thorup--Zwick distance oracle~\cite{thorup2005approximate}. This technical connection may be of independent interest, considering the wide application of~\cite{thorup2005approximate} in various computational settings.
\end{itemize}

We emphasize that most of our algorithms are \emph{incremental} in the sense that for any given parameter~$k$, they return a sequence of centers such that any prefix of length $ \ell \leq k $ yields an $\ell$-clustering solution with a constant-factor approximation.
\end{abstract}

\thispagestyle{empty}
\newpage

\tableofcontents
\thispagestyle{empty}
\newpage

\section{Introduction} \label{sec:introduction}
Clustering is a fundamental task in data mining and unsupervised learning  with numerous applications across various domains~\cite{JainMF99,HansenJ97,Berkhin06,AggarwalR2014}.
Intuitively speaking, the goal of a clustering algorithm is to group the input data together, in structures called ``clusters'', such that similar data points are contained in the same cluster.
    
A very popular, and arguably also very intuitive, clustering paradigm is \emph{partitional} clustering in which the clusters are represented by a set of \emph{centers} of a prespecified size $ k $ and each point in the data set is assigned to the cluster of its closest center.
The quality of a solution is measured based on an objective on the distances of points to their closest center.

Formally, this type of $ k $-clustering can be formulated as an optimization problem in which, for a given metric space $ \mathcal{M} = (P, \dist) $ and a given integer $ k \geq 1 $, the goal is to find a set of centers $ C \subseteq P $ with $ |C| = k $ minimizing one of the following objective functions $ f_\mathcal{M} $.
\begin{itemize}[itemsep=0.5em]
    \item $k$-center objective: $ f_\mathcal{M} (C) = \max_{p \in P} \dist (p, C) $, where $ \dist (p, C) = \min_{q \in C} \dist(p, q)$. Here, the goal is to minimize the maximum distance of any point to its closest center.
    \item $(k, z)$-clustering objective: $ f_\mathcal{M} (C) = \sum_{p \in P} \dist (p, C)^z $ for a given constant $ z \geq 1 $, where $\dist (p, C) = \min_{q \in C} \dist(p, q)$. Here, the goal is to minimize the sum of the distances from each point to its nearest center, each raised to the power~$z$.
    For $ z = 1 $ this is known as the \emph{$k$-median objective}, and for $ z = 2 $ this is known as the \emph{$k$-means objective}.
\end{itemize}

Clustering under these objectives is an NP-hard problem~\cite{WilliamsonS11} and therefore research has mainly focused on efficient approximation algorithms.
There are two main algorithmic research directions: 
\begin{enumerate}[itemsep=0.5em]
    \item Obtaining polynomial-time algorithms with the smallest possible approximation ratio~\cite{Bartal96,CharikarCGG98,CharikarGTS02,CharikarG05,JV01,JainMS02,JainMMSV03,mettu2003online,GuptaT08,CharikarL12,LiS16,BPR+17,AhmadianNSW20,Cohen-AddadGHOS22,GowdaPST23,Cohen-Addad0LS23,Cohen-Addad0LSS25,CharikarCGGLW25}.\antonis{Cite the new $k$-median, $k$-means algorithms.}

    \item Obtaining constant-factor approximation algorithms with the smallest possible running time~\cite{Indyk99,GuhaMMO00,BadoiuHI02,MettuP04,thorup2005quick,Chen09,CostaF25}.\antonis{I think that~\cite{CostaF25} has no $O(1)$-approx but $\log(\cdot)$.}
\end{enumerate}
In the latter direction, most notably, algorithms with running time $ \tilde O (k n) $ have been developed~\cite{Gon85,HS85,MettuP04,thorup2005quick} (assuming constant-time query access to distances between points of the metric space).
This matches (up to polylogarithmic factors) distance-query lower bounds of~$\tilde \Omega (k n) $~\cite{GuhaMMO00,MettuP04,BateniEFHJMW23} in this setting.

Some of these algorithms solve a variant of the problem known as \emph{incremental}\footnote{Note that in this context ``incremental'' does \textbf{not} refer to the algorithm being dynamic in the sense that data might be added to the input over time. Later in this paper, we will actually employ an algorithm that is incremental in the dynamic-algorithms sense. To avoid confusion, we will avoid the term ``incremental'' for the latter.} $k$-clustering~\cite{mettu2003online,LinNRW10}, in which the goal is to output a sequence $ c_1, \ldots, c_k $ of centers such that, for every $ 1 \leq \ell \leq k $, the subset $ c_1, \ldots, c_\ell $ is an (approximate) $ \ell $-clustering solution.
Incremental clustering gives an efficiency gain in uncertain situations where only an upper bound on the desired number of clusters is known and the results for smaller values can instantly be retrieved without running the algorithm from scratch each time.

To obtain faster running times beyond the $ kn $ barrier, researchers have studied $k$-clustering in metric spaces with more structure.
Two such examples are Euclidean spaces~\cite{FederG88,Har-Peled04,latourSaulpicKz,jiang2026local} (with the Euclidean distance between points) and graphs of low-enough density~\cite{thorup2001quick,abboud2023fine,latourSaulpicKz,jiang2026local} (with the induced shortest-path distance between vertices).
In this paper, we focus on the setting in which the metric is a graph.
This is motivated by network structures such as road networks, computer networks, or social networks, which appear naturally in many contexts where clustering is employed to identify structures like communities.
Furthermore, techniques for turning metric input data into a graph~\cite{Luxburg07} are sometimes employed to make large amounts of data amenable to analysis.
This paradigm, in the form of first computing a sparse Euclidean spanner~\cite{arya1995euclidean,Har-PeledIS13,ElkinS15,le2025truly} and then running a graph algorithm on the spanner, has recently also been applied to obtain the best theoretical guarantees for Euclidean $ (k, z) $-clustering in a certain regime~\cite{jiang2026local}.
In any case, the input to the algorithm is often a massive sparse graph, and even moderately super-linear running times in the number of edges $m$ are quite prohibitive.
Thus, it is natural to seek algorithms that are as close as possible to a linear (in number of edges) running time.

While this linear-time ``gold standard'' has been approached for the $k$-clustering objectives mentioned above~\cite{thorup2001quick,abboud2023fine,latourSaulpicKz,jiang2026local}, there are still gaps concerning (a) how close we can get to linear time and (b) deterministic algorithms achieving such fast running times.
The goal of this work is therefore to narrow these gaps in approaching linear running time for variants of $k$-clustering problems on graphs.

On the technical side, our contribution is to utilize ideas from distance approximation algorithms in graphs in the $k$-clustering context.
In particular, we (a) employ a dynamic shortest-path approximation algorithm~\cite{gorkiewicz2025incremental} and (b) utilize techniques from Thorup--Zwick distance oracles~\cite{thorup2005approximate, thorup2001compact, roditty2005deterministic}.
Dynamic shortest-path algorithms in general have proved to be useful in many algorithmic settings where shortest paths need to be computed repeatedly in an ``outer loop'' that slightly modifies the graph with each iteration (see \cite{RodittyZ11} and~\cite{Madry10} for two prominent examples in context of flows and spanner). We identify this behavior for $k$-center and $(k,z)$-clustering and leverage the power of dynamic approximate shortest paths to obtain efficiency gains.
For the $(k, z)$-clustering problem,\antonis{this refers only to the deterministic section? If so, then we should mention it.} we utilize the data structures used in the distance oracle construction of Thorup and Zwick~\cite{thorup2005approximate, thorup2001compact, roditty2005deterministic}. By adapting these structures, we can replace a randomized method of estimating sizes of balls (up to a certain radius) around each node~\cite{cohen1997size} by a deterministic one with mild overheads.
%

\subsection{Our Results and the State of the Art}
The following characterization of running-time statements will facilitate the discussion of our and prior work.
We say that the worst-case \antonis{worst-case is the default assumption; should we remove it to avoid confusion?} running time $ T (n, m) $ of a given algorithm on any graph with $ n $ nodes and $ m $ edges is: 
\begin{itemize}
    \item \textbf{Nearly-linear}, if $ T (n, m) = O (m \cdot (\log n)^c) $ for some constant $ c $ (``polylogarithmic overhead'').
    \item \textbf{Almost-linear}, if $ T (n, m) = O (m^{1+o(1)}) $ (``subpolynomial overhead'').
    \item \textbf{Close-to-linear}, if $ T (n, m) = O (m^{1+\mu}) $ for any constant $ \mu > 0 $ (``arbitrarily-small-polynomial overhead''). 
\end{itemize}

In the following, we outline our results and how they compare to the state of the art. A quick overview is also summarized in Table~\ref{tab:results}.

\definecolor{newresult}{RGB}{230, 245, 220}   

\begin{table}[t]
\centering
\renewcommand{\arraystretch}{1.35}
\caption{%
  Summary of new results and the most relevant state of the art for $k$-center and $(k,z)$-clustering on
  weighted undirected graphs $G=(V,E,w)$ with $n=|V|$, $m=|E|$, and aspect ratio 
  $\Delta$ .
  The $\tilde{O}(\cdot)$ notation hides $\operatorname{polylog}(n,\Delta)$ factors.
  Rows highlighted in \colorbox{newresult}{\strut green} are contributions of this paper.
}

\label{tab:results}
\setlength{\tabcolsep}{6pt}
\begin{tabular}{lllll}
\toprule
\textbf{Problem}
  & \textbf{Approx.}
  & \textbf{Running time}
  & \textbf{Det./Rand.}
  & \textbf{Reference} \\
\midrule

\multirow{2}{*}{$k$-center}

  & $2$ & $\Tilde{O}(m)$   & Rand.
    & \cite{thorup2001quick} \\
  & $2+\varepsilon$ & $\Tilde{O}(m/\log(\varepsilon))$   & Rand.
    & \cite{abboud2023fine} \\
\cmidrule{2-5}
  & \cellcolor{newresult}$2+\varepsilon$ & \cellcolor{newresult}$\tilde{O}(m/\varepsilon)$       & \cellcolor{newresult}Det.
    & \cellcolor{newresult}\Cref{lem:deterministic_k_center_gon} \\
\midrule

\multirow{3}{*}{\shortstack[l]{$k$-center,\ $t$ outliers \\ (constant $k$)}}
  & $3^{\ddagger}\ $  & $\mathrm{poly}(n)$  & Det.   & \cite{CKMN01} \\
  & $2^{\dagger}$  & $\mathrm{poly}(n)$  & Rand.  & \cite{ding2019greedy} \\
\cmidrule{2-5}
 \cmidrule{2-5}
  & \cellcolor{newresult}$2+\varepsilon ^{\dagger}$ & \cellcolor{newresult}$\tilde{O}(m/\varepsilon)$       & \cellcolor{newresult}Rand.
    & \cellcolor{newresult}\Cref{lem:outlier-2approx} \\
 
\midrule
 
\multirow{2}{*}{\shortstack[l]{$k$-center, $t$ outliers\\ (general $k$) }}
  & $(2,\,O(1/\varepsilon))^{*}$ & $\mathrm{poly}(n)$ & Rand.
    & \cite{ding2019greedy} \\
\cmidrule{2-5}
  \cmidrule{2-5}
  & \cellcolor{newresult}$(2+\varepsilon, O(1/\varepsilon))^{*}$ & \cellcolor{newresult}$\tilde{O}(\frac{m}{\varepsilon})$       & \cellcolor{newresult}Rand.
    & \cellcolor{newresult}\Cref{lem:outlier-bicriteria} \\

\midrule


$k$-median ($ z = 1 $)  & $\mathrm{O}(1)$  & $ \tilde O (m) $          & Rand.
    & \cite{thorup2001quick} \quad \\
\multirow{3}{*}{$(k,z)$-clustering}  & $\mathrm{O}(c^6)$  & $\Tilde{O}(m^{1+\frac{1}{c}})$          & Rand.
    & \cite{latourSaulpicKz} \quad \\
  & $\mathrm{O}(1)$  & $\Tilde{O}(m^{1+o(1)})$          & Rand.
    & \cite{jiang2026local} \quad  \\
   
\cmidrule{2-5}
  & \cellcolor{newresult}$O(1)$   & \cellcolor{newresult}$\tilde{O}(m)$     & \cellcolor{newresult}Rand.
    & \cellcolor{newresult}\Cref{thm:randomizedkzclustering} \\
  & \cellcolor{newresult}$O(\text{poly(c)})$   & \cellcolor{newresult}$\Tilde{O}(mn^{\frac{1}{c} } + \frac{m}{\varepsilon})$     & \cellcolor{newresult}Det.
    & \cellcolor{newresult}\Cref{thm:deterministicclustering} \\
 
\bottomrule
\end{tabular}
 
\smallskip
\noindent\footnotesize
$^\dagger$ : uses at most $O(1+\epsilon)t$ outliers
$\ddagger$ : $k$ is not a constant 
* - These are Bi-criteria Algorithms
\end{table}

\subparagraph{$k$-Center clustering problem.}

A randomized nearly-linear time $ 2 $-approximation algorithm for $ k $-center on graphs has been developed by Thorup~\cite{thorup2005quick}. 
Relaxing the approximation factor to $ 2 + \epsilon $, Abboud, Cohen-Addad, Lee, and Manurangsi~\cite{abboud2023fine} have obtained (1) a slightly faster randomized nearly-linear time algorithm and (2) a randomized nearly-linear time algorithm  for the incremental $k$-center problem.\antonis{Is this the Gonzalez algorithm?}

The fastest deterministic algorithms are, to the best of our knowledge, straightforward adaptations of the classic $2$-approximation algorithms by Gonzalez~\cite{Gon85} or Hochbaum and Shmoys~\cite{HS85}, which run in time $ \tilde O (k m) $. The former also solves the incremental $k$-center problem.
In addition, following an observation by Thorup~\cite{thorup2005quick}, a deterministic incremental constant-factor approximation algorithm for $ k $-center clustering with running time $ \tilde O (n^2) $ can be obtained by first running the approximate all-pairs shortest paths algorithm of Cohen and Zwick~\cite{CohenZ01} and then running Gonzalez's algorithm~\cite{Gon85} with the precomputed approximate distances.
Designing a deterministic nearly-linear time $k$-center algorithm was explicitly stated as an open problem by Abboud et al.~\cite{abboud2023fine}.

We resolve this open problem and provide a $ (2 + \epsilon) $-approximation algorithm that also works for the incremental version.
More formally, we obtain the following theorem, which we prove in \Cref{sec:kcenter}.
\begin{restatable}[]{thm}{kcenter}
\label{lem:deterministic_k_center_gon}
There is a deterministic algorithm for the incremental $k$-center problem that, given a weighted undirected graph $G = (V,E,w)$, an integer $ k \geq 1 $, and an accuracy parameter $ 0 < \varepsilon \leq 1 $, computes a $(2 + \epsilon)$-approximate solution in time $\Tilde{O}(\frac{m}{\varepsilon})$.
\end{restatable} \antonis{If the graph is not connected, an additional factor $k$ appears in the runtime. In the analysis, do we argue something like $k \leq m$?}

L{\'a}szl{\'o} Kozma~\cite[Theorem~1.1]{kozma2026price} proved that no incremental
$k$-center algorithm, regardless of its running time, has approximation
ratio strictly below~$2$ on every prefix, already on weighted path
graphs. Hence \Cref{lem:deterministic_k_center_gon} is optimal for incremental algorithms
up to the additive~$\varepsilon$.

An interesting variant of $k$-center providing more robustness in real-world data analysis tasks allows the algorithm to specify up to $ t $ outliers (that then are not evaluated in the objective function).
For this $k$-center with outliers problem, previous works by Charikar et al. \cite{CKMN01} and Ding et al (only for constant $k$ and having $(1+\epsilon)t$ outliers). \cite{ding2019greedy} provided polynomial-time algorithms with approximation factors of 3 (deterministic) and 2 (randomized), respectively.

We provide a randomized nearly-linear time $ (2 + \epsilon) $-approximation algorithm for $k$-center clustering with outliers that essentially is an efficient implementation of the Ding et al.\cite{ding2019greedy} approach.
More formally, we obtain the following theorem that we prove in \Cref{sec:kcenter}.

\begin{restatable}[]{thm}{kcenteroutlier}
\label{lem:outlier-2approx}
    For constant $ k $ \footnote{This is the only result in the paper where we assume $k$ is a constant. All other statements hold for any integer $k$.}, there is a randomized algorithm for the $k$-center problem with at most $t $ outliers that, given a weighted undirected graph $G = (V,E,w)$, integer $ t \geq 0 $, and an accuracy parameter $ 0 < \varepsilon \leq 1 $, computes with probability $
  (1 - t/n) \cdot \left( \frac{\epsilon}{\lceil( 1+\varepsilon)\rceil} \right)^{k-1},
$ a $(2 + \epsilon)$-approximate solution in time $\Tilde{O}(\frac{m}{\varepsilon})$ with at most $(1+\epsilon)t$ outliers .
\end{restatable}


In Appendix~\Cref{sec:outliers}, we also provide an efficient implementation of a bi-criteria approximation for $k$-center with outliers that achieves constant-factor approximation while allowing $O(\frac{k}{\varepsilon})$ centers.

\subparagraph{$(k, z)$-Clustering problem.}

A randomized nearly-linear time constant-factor approximation algorithm for $ k $-median ($ z = 1 $) was developed by Thorup~\cite{thorup2005quick}.
Recently, two independent works have addressed fast $ (k, z)$-clustering on graphs, which in particular includes solutions for $k$-median ($ z = 1 $) and $k$-means ($ z = 2$).
First, Jiang, Jin, Lou, and Lu~\cite{jiang2026local} have presented a randomized almost-linear time algorithm for $ (k, z) $-clustering with constant-factor approximation.
Second, Dupré la Tour and Saulpic~\cite{latourSaulpicKz} have presented a randomized close-to-linear time algorithm for incremental $ (k, z) $-clustering with constant-factor approximation.
As remarked by~\cite{latourSaulpicKz}, it is conceivable that Thorup's approach may be extendable to $k$-means using recent results on the primal-dual method for $k$-means.
However, Thorup's approach does not readily extend to incremental clustering and it is arguably much more complex than the greedy framework of~\cite{latourSaulpicKz}.

We generalize/subsume all of these results by providing a randomized nearly-linear time algorithm for incremental $ (k, z) $-clustering with constant-factor approximation.
More formally, we obtain the following theorem that we prove in \Cref{sec:randomized}. 

\begin{restatable}[]{thm}{randomizedkzclustering}
\label{thm:randomizedkzclustering}
    There is a randomized algorithm for the incremental\footnote{Recall that \emph{incremental} does not mean the algorithm is dynamic; rather, it means that every prefix of length $\ell$ of the returned $(k, z)$-clustering solution induces a constant-factor approximate $(\ell, z)$-clustering solution.} $(k, z)$-clustering problem that, given a weighted undirected graph $G = (V,E,w)$, an integer $ k \geq 1 $ aspect ratio $\Delta$ \footnote{The aspect-ratio is the ratio between the largest distance and the smallest non-zero distance. }, and a constant $z \geq 1$, computes with high probability an $O(1)$-approximate solution in $\Tilde{O}(m \log(\Delta))$ time.
\end{restatable}
    
To the best of our knowledge, no deterministic algorithm tailored to the graph setting is known.
However, again following the observation by Thorup~\cite{thorup2005quick}, a deterministic incremental constant-factor approximation algorithm for $ (k, z) $-clustering with running time $ \tilde O (n^2) $ can be obtained by first running the approximate all-pairs shortest paths algorithm of Cohen and Zwick~\cite{CohenZ01} and then running the metric incremental $ (k, z) $-clustering algorithm of Mettu and Plaxton~\cite{mettu2003online} with the $\rho$-metric (obeying the triangle inequality up to a factor $ \rho = 3 $) given by the approximate shortest-path distances.

We resolve this limitation by providing the \emph{first deterministic} close-to-linear time algorithm for incremental $ (k, z) $-clustering with constant-factor approximation.
More formally, we obtain the following theorem, which we prove in \Cref{sec:deterministic}.

\antonis{Should we write explicitly in the statement that we have the ``incremental/nested'' version?}

\begin{restatable}[]{thm}{deterministicclustering}
\label{thm:deterministicclustering}
There is a deterministic algorithm for the incremental
$(k,z)$-clustering problem that, given a weighted undirected graph
$G = (V,E,w)$, an integer $k \ge 1$, an integer parameter $t \ge 1$, a
constant $z \ge 1$, and an accuracy parameter $0 < \varepsilon \le 1$,
computes an $O(\mathrm{poly}(t))$-approximate solution in time
$\tilde{O}\bigl(t \cdot m \cdot n^{\frac{1}{t}} +
\tfrac{m}{\varepsilon}\bigr)$.
\end{restatable}
\antonis{Do we need $O(\cdot)$ or just say $\poly(c)$-approximate?}

\section{Technical Overview}
\label{sec:overview}

All of our algorithms utilize a deterministic Source-Insertion dynamic $(1+\epsilon)$-SSSP (single-source shortest paths) algorithm, recently developed by Górkiewicz and Karczmarz~\cite{gorkiewicz2025incremental}. This data structure achieves $\tilde{O}(m)$ total update time for processing adversarial edge insertions only to the source.\footnote{The data structure of \cite{gorkiewicz2025incremental} avoids a usual $n^{o(1)}$-factor appearing in partially dynamic $(1+\epsilon)$-SSSP data structures by relying on the fact that edge insertions are only incident to the source.} Our deterministic $k$-center and $k$-center with outlier results are mainly based on a direct application of this data structure in the Gonzalez framework \cite{Gon85}. Namely, we use a Source-Insertion SSSP to iteratively find the furthest vertex to the current set of centers.

Our $(k,z)$-clustering algorithms are based on the recursive greedy algorithm by Dupr\'{e} la Tour and Saulpic~\cite{latourSaulpicKz} that gives an efficient implementation of the Mettu-Plaxton \cite{mettu2003online} framework. For a vertex $v \in V$ and radius $r \ge 0$, the ball of radius $r$ around $v$ is defined as
$
B({v},{r}) := \{ u \in V \mid \dist(v,u) \le r \}.$
At a high level, the framework iteratively uses two main primitives:
\begin{itemize}[itemsep=0.5em]
    \item \emph{Ball-size estimation}: estimating $|\Ball{v}{r}|$ for specific
    $(v,r)$ pairs, .
    
    \item \emph{Computing approximate balls}: compute an approximate ball $\NBall{v}{r}$ such that $\Ball{v}{r}\subseteq\NBall{v}{r}\subseteq\Ball{v}{c \cdot r}$, where $c$ is a constant. 
\end{itemize}

We first explain how using the data structure of \cite{gorkiewicz2025incremental} implies an improved randomized result. Then we provide a more technical deterministic algorithm that uses primitives from the Thorup-Zwick distance oracles to efficiently compute and query approximate neighborhood. More specifically: 
\subparagraph{Randomized approach.}
For the ball-size estimation primitive, similar to \cite{latourSaulpicKz}, we plug in the ball-size estimation data structure by Cohen~\cite{cohen1997size}. 
In the recursive greedy framework, we need to repeatedly and greedily select a set of centers and approximate balls. Our novelty in ensuring that all of this iterative process can be done in  nearly-linear time is the use of two shortest-path primitives.

One primitive uses a truncated Dijkstra from a set of carefully selected centers in such a way that ensures each edge is relaxed \emph{at most twice} for any fixed radius $r$. This is repeated for $O(\log\Delta)$ radius scales, leading to $O(m\log\Delta)$ time. 

Another primitive is needed to compute approximate balls around the set of chosen centers efficiently.
Recomputing the approximate balls for each selected center would be too expensive. Therefore we amortize these computations via the Source-Insertion SSSP of \cite{gorkiewicz2025incremental} in total linear time. 
In the variant of~\cite{latourSaulpicKz} for graphs, this step uses probabilistic graph decompositions from~\cite{filtser2019strong}, which incurs an additional polynomial factor that we manage to avoid.

\subparagraph{Deterministic approach.}
There are two source of randomness in the algorithm of \cite{latourSaulpicKz}: one is Cohen's ball estimation \cite{cohen1997size}, and another one is an iterative approximate balls computation. In \cite{latourSaulpicKz}, this is done via a probabilistic decomposition. We make the second primitive deterministic with the described Source-Insertion data structure and truncated Dijkstra. However, a major challenge is using the ball-size estimation. We need a new approach for making this deterministic as Cohen's algorithm \cite{cohen1997size} is inherently randomized.

We replace this step with a deterministic ball approximation approach based on the notion of bunch and clusters used in the Thorup-Zwick distance oracle~\cite{thorup2005approximate}. 
Concretely, we define $\NBall{v}{r}$ as the union of the truncated bunch $\mathrm{Bunch}(v,r)$ with the truncated clusters of $v$'s pivots
(\Cref{eq:approx_ball_bunch}). Here additional approximation and running time factors are added compared to our randomized construction, that depend on the balls approximation factors.

We have to be careful in how we define the approximate balls, so that each vertex overlaps with a constant number of such neighborhoods. A more naive use of Thorup-Zwick clusters results in each vertex overlapping with $n^{1/c}$ clusters, for some constant $c$. However, we define the neighborhoods in such a way that they overlap with $c$ such clusters, and this is crucial for our final approximation factor.

\section{Deterministic Greedy $k$-Center and Faster $k$-Center with Outlier} \label{sec:kcenter}

In this section, first, we will state the Source-Insertion SSSP \Cref{thm:simplified-sssp}, which will be used in the $k$-center \Cref{alg:gonzalez_incremental} and also later in $(k,z)$  clustering \Cref{thm:randomizedkzclustering}. Next, we give a brief overview of Gonzalez algorithm and the modification required for our result. Lastly, we will sketch how it can be extended to $k$-center with outlier.  

\begin{lemma}[Source-Insertion SSSP~\cite{gorkiewicz2025incremental}] 
\label{thm:simplified-sssp}
Consider a parameter $\varepsilon \in (0,1)$. Let $G = (V, E, w)$ be a directed graph with edge weights in
$\{0\} \cup [1,W]$, and let $s \in V$ be a source vertex.
There exists a deterministic data structure that explicitly maintains distance estimates
$\delta : V \to \mathbb{R}_{\ge 0}$ satisfying~$\operatorname{dist}_G(s,v) \;\le\; \delta(v) \;\le\; (1+\varepsilon)\cdot \operatorname{dist}_G(s,v) \;\; \text{for all vertices } v \in V$.

The data structure supports only insertions (or weight decreases) of source edges
$e = sv$, with~$v \in V$. The total update time is~$O\!\left(\frac{m \log(\Delta)\,\log^2 n}{\varepsilon} + \Lambda \right)$,
where $m$ is the final number of edges in $G$, and $\Lambda$ is the total number of updates. \antonis{Let's replace $\Delta$ here with something else, since we use $\Delta$ for the aspect ratio.}
\rajath{yes, by $\Lambda$}
\end{lemma}

The classical greedy algorithm for $k$-center, due to Gonzalez~\cite{Gon85}, proceeds as follows. It starts by selecting an arbitrary vertex $c_1 \in V$ as the first center. Let $C_1 = \{c_1\}$. In each subsequent iteration $i = 2, \dots, k$, the algorithm selects the vertex $c_i$ that is farthest from the current set of centers $C_{i-1}$:
$ c_i \gets \max_{v \in V} \dist(v, C_{i-1}), $ and updates the set of centers as $C_i \gets C_{i-1} \cup \{c_i\}$. This algorithm guaranties a $2$-approximation for the $k$-center problem.

Our deterministic algorithm builds on the classical Gonzalez algorithm \cite{Gon85}, but achieves $\Tilde{O}(m)$ time complexity through careful use of Source-Insertion Single Source Shortest Path computations. The key insight is that maintaining distances from a super-source vertex connected to all selected centers allows us to efficiently identify the (approximate) farthest vertex from the current set of centers at each iteration.

\kcenter*
\begin{proof}
Let $G = (V,E,w)$ be a graph with nonnegative edge weights. We implement the classical Gonzalez \cite{Gon85} (see \Cref{alg:gonzalez_incremental}) using the $(1+\varepsilon)$-approximate Insert-only Single Source Shortest Path data structure of ~\Cref{thm:simplified-sssp}. Let $\hat{d}_s(\cdot)$ denote the maintained distance estimate from $s$.


\begin{algorithm}[h]

\caption{Deterministic $(2+\varepsilon)$-Approximate $k$-Center}
\label{alg:gonzalez_incremental}

\KwIn{Weighted graph $G=(V,E,w)$, integer $k$}
\KwOut{Set of centers $C$}

Construct $ G' = (V \cup \{ s \}, E, w )$\;

Initialize the $(1+\varepsilon)$-approximate incremental SSSP structure from $s$\;

Pick an arbitrary vertex $c_1 \in V$\;

Insert edge $(s,c_1)$ with weight $0$\;

$C \leftarrow \{c_1\}$\;

\For{$i = 2$ \KwTo $k$}{

    $c_i \leftarrow \arg\max_{v \in V} \hat d_s(v)$\;
    
    Insert edge $(s,c_i)$ with weight $0$\;
    
    $C \leftarrow C \cup \{c_i\}$\;
}

\Return $C$\;

\end{algorithm}

\Cref{alg:gonzalez_incremental} is implemented using the
Source-Insertion SSSP data structure of \Cref{thm:simplified-sssp} as
follows. At each iteration $i = 2 \ldots k$, the algorithm inserts a zero-weight edge $(s,c_i)$.  This causes the Source-Insertion SSSP structure to update the distance estimates $ \operatorname{dist}(v,C_i) \le \hat{d}(v)\leq (1+\varepsilon) \operatorname{dist}(v,\{c_1,\ldots,c_i\})$ for all $v\in V$. More formally,

\begin{observation}\label{obs:distance_estimate}
After processing center set $C_i = \{c_1, \ldots, c_i\}$ through the incremental SSSP 
data structure, the maintained distance estimates satisfy:
$
    \dist_G(v, C_i) \leq \hat{d}(v) \leq (1+\varepsilon) \cdot \dist_G(v, C_i)
    \quad \text{for all } v \in V,
$
where $\hat{d}(v)$ is the distance estimate from $s$ to $v$ in the augmented graph $G'$.
\end{observation}

To implement the $\arg\max_{v \in V} \hat{d}_s(v)$ in Line $7$ of the \Cref{alg:gonzalez_incremental} at each iteration, we maintain the values $\hat{d}_s(v)$ in a binary heap (or any other priority queue) keyed by the current distance estimate. Whenever a distance estimate decreases, we perform a corresponding decrease-key operation in the heap.

\begin{restatable}[]{lem}{kcenterproof}
\label{thm:det-gonzalez}
~\Cref{alg:gonzalez_incremental} produces an incremental $(2+\varepsilon)$-approximation for the $k$-center problem.
\end{restatable}

\begin{proof}
Let $C_i = \{c_1, \dots, c_i\}$ be the center set after iteration $i$ and let $\hat{d}_i(v)$ be the distance estimate after processing $i$ centers. Define $c_{k+1} \gets \arg\max_{v \in V} \hat{d}_k(v)$ as a hypothetical $(k+1)$-th center, and let $r = \hat{d}_k(c_{k+1})$.
By \Cref{obs:distance_estimate}, $\dist_G(v, C) \leq \hat{d}_k(v) \leq r$ for all $v \in V$, so it suffices to bound $r$.

For each $i = 2, \dots, k+1$, the algorithm chose $c_i$ as the approximate 
farthest point from $C_{i-1}$, so by \Cref{obs:distance_estimate},
$
    \dist_G(c_i, C_{i-1}) 
    \;\geq\; \frac{\hat{d}_{i-1}(c_i)}{1+\varepsilon} 
    \;\geq\; \frac{r}{1+\varepsilon},
$
where the last inequality uses the fact that distance estimates are 
non-increasing as centers are added, so $\hat{d}_{i-1}(c_i) \geq r$. 
In particular, all $k+1$ points $c_1, \dots, c_{k+1}$ are pairwise at 
distance greater than $r/(1+\varepsilon)$.

Since the optimal solution has only $k$ centers, two of the $k+1$ points 
$c_i, c_j$ must share the same optimal center. By the triangle inequality, $$
    \frac{r}{1+\varepsilon} 
    \;<\; \dist_G(c_i, c_j) 
    \;\leq\; 2 \cdot \mathrm{OPT}.$$

Rearranging gives $\max_{v \in V} \dist_G(v, C) \leq r \leq 2(1+\varepsilon) 
\cdot \mathrm{OPT}$. Running the insert only SSSP structure with $\varepsilon/2$ in place 
of $\varepsilon$ yields the claimed $(2+\varepsilon)$-approximation. The incremental property holds because for any $i \leq k$, the sequence $c_1, \dots, c_i$ is exactly the output of the algorithm run with parameter $i$ instead of $k$. The same analysis applies to show that $C_i$ is a $(2+\varepsilon)$-approximation compared to $\mathrm{OPT}_i$, where $\mathrm{OPT}_i$ denotes the optimal radius achievable with $i$ centers..

\end{proof}

We will use the below \Cref{lem:inc-sssp-output} in the proof of our running time \Cref{lem:det-gon-time}. 

\begin{lemma}[Vertex‑touch reporting~\cite{gorkiewicz2025incremental}] \label{lem:inc-sssp-output}
The incremental $(1+\varepsilon)$‑approximate SSSP structure of~\Cref{thm:simplified-sssp} for any source‑edge insertion $(s,v)$ returns a set
$V_{\rm touched}\subseteq V$ consisting of exactly the vertices whose
distance estimate decreases because of the insertion.
Moreover, over any sequence of $\Lambda$ such insertions the total running time
(and the total size of all returned sets) is
$
O\!\bigl(m\log(\Delta)\log n/\varepsilon+\Lambda\bigr)=\tilde O(m)
$
when $\Lambda=O(n)$ and $\varepsilon=\Theta(1)$.
\end{lemma}

\begin{restatable}[]{lem}{kcentertime}
\label{lem:det-gon-time}
    The running time of \Cref{alg:gonzalez_incremental} is $\Tilde{O}(\frac{m}{\varepsilon})$
\end{restatable}

\begin{proof}
    We account for each component of the algorithm separately.

    \textbf{Insert only SSSP structure.} Over all $k$ iterations, exactly $k \leq n$ 
    zero-weight edges $(s, c_i)$ are inserted into the augmented graph $G'$. 
    By ~\Cref{thm:simplified-sssp}, the Insert only SSSP structure 
    handles the full sequence of insertions and maintains all distance 
    estimates in total time $\tilde{O}(\frac{m}{\varepsilon})$. \\ 
    \textbf{Heap operations.} Since each inserted edge has weight zero and 
    connects $s$ directly to a center, inserting $(s, c_i)$ can only 
    \emph{decrease} distance estimates --- never increase them. Therefore, 
    the heap requires only \texttt{decrease-key} and \texttt{find-max} 
    operations throughout the algorithm. The number of \texttt{decrease-key} 
    operations equals the number of distance updates performed by the Insert only SSSP 
    structure, which is $\tilde{O}(m)$ by \Cref{lem:inc-sssp-output}. 
    Each such operation costs $O(\log n)$, contributing $\tilde{O}(m)$ in 
    total. The $k \leq n$ \texttt{find-max} calls, each cost 
    $\Tilde{O}(1)$. 
    
    Both components run in at most $\tilde{O}(\frac{m}{\varepsilon})$, hence the total 
    running time of ~\Cref{alg:gonzalez_incremental} is $\tilde{O}(\frac{m}{\varepsilon})$.
\end{proof}
\end{proof}
\begin{remark}
An alternative approach constructs a maximal distance-$d$ independent set via
the same Source-Insertion SSSP structure, then binary-searches over $d$ to achieve
a $(2+\varepsilon)$-approximation.  This also runs in $\tilde{O}(m)$ time but does not naturally yield an \emph{nested} center sequence, so we prefer the Gonzalez approach here.
\end{remark}

For the $k$-center with $t$ outliers, in \Cref{alg:gonzalez_incremental}, at each iteration $i$, we find a set $S$ of the $(1+ \varepsilon)t$ furthest vertices from the super source and randomly select one vertex from the set $S$ as the next center. We prove \Cref{lem:outlier-2approx} in \Cref{sec:outliers}.

\section{Greedy $(k,z)$-Clustering Algorithm} \label{sec:greedy_kz}

\antonis{Should we mention that some of the parameters (e.g., $c^4$) are an exaggeration and come from the other paper?}
In this section we review the greedy $(k,z)$-clustering algorithm of Dupré la Tour and Saulpic~\cite{latourSaulpicKz}, a simplification of
the algorithm of Mettu and Plaxton~\cite{mettu2003online}, and we
present it as a \emph{template}. The algorithm is stated in terms of
two abstract primitives -- ball values and approximate balls  (\cref{eq:value_ball,eq:approx_ball} below), not on
how they are implemented. Both of our $(k,z)$-clustering results are
obtained by instantiating this template. In \cref{sec:randomized}, we
give a randomized implementation -- ball values via Cohen's ball-size
estimation~\cite{cohen1997size}, approximate balls via truncated
Dijkstra computations, and a forbidding loop driven by a single
Source-Insertion SSSP structure -- which yields
\Cref{thm:randomizedkzclustering}. In \Cref{sec:deterministic}, we give
a deterministic implementation -- ball values
via the bunches and clusters of the Thorup--Zwick
hierarchy~\cite{thorup2005approximate} -- which yields \Cref{thm:deterministicclustering}.

 The pseudo-code of this simplified greedy algorithm is provided in~\cref{alg:simplified-greedy}. Similarly to~\cite{latourSaulpicKz}, we use a parameter $c \geq 5$ and assume for simplicity that the aspect ratio $\Delta$ is a power of $2c$. 

\antonis{I replaced ``Let $B \coloneqq B(u,r)$ be a ball'', since that $B$ was not used anywhere, right?}
\begin{algorithm}[h]
\caption{Simplified Recursive Greedy}
\label{alg:simplified-greedy}
\KwIn{A weighted undirected graph $G = (V, E, w)$ and an integer $k \geq 1$}
\KwOut{A set $C$ of at most $k$ centers}
\DontPrintSemicolon
\vspace{0.5em}
Define the set of available balls as: $\mathcal{B} \coloneqq \bigl\{B\bigl(u, \frac{\Delta}{(2c)^\ell}\bigr) \,\big|\, u\in V,\; \ell\in\{0,\ldots,\log_{2c}(\Delta)+7\} \bigr\}$
\vspace{0.2em}

Compute $\mathrm{Value}(B(u,r))$ for all available balls
\vspace{0.4em}

\For{$i=1$ \KwTo $k$}{
    \If{$\mathcal{B}$ is \text{empty}}{
        \textbf{output} $C_{i-1}=\{c_1,\ldots,c_{i-1}\}$
    }
    Let $B(u,r)$ be a ball in $\mathcal{B}$ with maximum $\mathrm{Value}(\cdot)$ \label{algline:B_init_cent_sel}
    \vspace{0.2em}
    
    \tcp{Center-selection loop}
    \While{$r > \bigl(\frac{1}{2c}\bigr)^7$} { 
        Compute an approximate ball $N \coloneqq N(u, 10c\cdot r)$ \label{line:compute_N_dijkstra}
        
        $u \gets \arg\max_{v\in N}\mathrm{Value}\Bigl(B\bigl(v, \frac{r}{2c}\bigr)\Bigr)$ \label{line:find_next_best_u}
        
        $r \gets \frac{r}{2c}$
    }

    $c_i \leftarrow u$

    \vspace{0.2em}
    \tcp{Forbidding loop}
    \ForEach{$r \in \bigl\{\frac{\Delta}{(2c)^\ell} \mid \ell \in \{0,\ldots,\log_{2c}(\Delta)+7\} \bigr\}$}{
        \ForEach{$u \in N(c_i,100c^4 \cdot r)$}{
            Remove the ball $B(u,r)$ from the set $\mathcal{B}$\;
        }
    }
}
\vspace{0.2em}
\textbf{Output} $C_k \coloneqq \{c_1,\ldots,c_k\}$\;
\end{algorithm}

Let $\mathcal{R} \coloneqq \left\{\frac{\Delta}{(2c)^\ell} \mid \ell \in \{0, \dots, \log_{2c}(\Delta) + 7\} \right\}$ be the set of radii.

\paragraph{The two primitives.} The template needs the graph for the  following two primitives (both of which tolerate
constant-factor slack).

(P1) For a parameter $\rho \ge 1$ with
$c \ge \rho^2/2$: for every $u \in V$ and $r \in \mathcal{R}$, the
algorithm requires a quantity $\mathrm{Value}(B(u,r))$ satisfying
\begin{equation}
  \frac{r^z}{\rho} \cdot |B(u, r)| \;\leq\; \mathrm{Value}(B(u, r))
  \;\leq\; \rho \myhs r^z \cdot |B(u, c \cdot r)|.
  \label{eq:value_ball}
\end{equation}

For our randomized $(k,z)$ clustering (\Cref{sec:randomized}), we take the value of $\rho = 3$ and for our deterministic algorithm (\Cref{sec:deterministic}) we take $\rho = t+1$, where $t$ is the number of levels of the
Thorup--Zwick distance oracle. The technical justification that any $\rho \geq 1$ and $c \geq \frac{\rho^2}{2}$ suffice is provided in ~\Cref{sec:approx-value}. 


\emph{(P2) Approximate balls.} For every $u \in V$ and radius $r$, the
algorithm requires a set $N(u, r)$ satisfying
\begin{equation}
  B(u, r) \;\subseteq\; N(u, r) \;\subseteq\; B(u, c \cdot r).
  \label{eq:approx_ball}
\end{equation}

    
\yasamin{I find this presentation strange: we give the algorithm, and then define some terms and then give high-level overview. I may bring the explanation earlier and bring the definitions where they become relevant}
 This greedy algorithm maintains a set of candidate balls of various radii and it proceeds in $k$ iterations, where $k \geq 1$ and $z \geq 1$ are the parameters of the $(k, z)$-clustering instance. 

In each iteration, the algorithm selects the ball with the highest ``value''---a function of its radius and the number of points it contains. The \emph{center-selection loop} is then performed, in which the algorithm recursively explores smaller-radius balls within the selected one to choose a precise center~$c_i$. Each vertex $u$ in the center selection loop from which we compute $N(u,r)$ is called candidate center. At the end of each iteration, the \emph{forbidding loop} is performed, in which the algorithm removes all candidate balls that are ``too close'' to the newly chosen center~$c_i$. Due to the forbidding loop, the center selection loop does not run from every vertex for a specified radius $r$. Helping us run the algorithm efficiently.  

The guarantee of this greedy algorithm is that each prefix $c_1, \ldots, c_\ell$ of the sequence of centers $c_1, \ldots, c_\ell, \ldots, c_k$ is an $O(1)$-approximation for the $(\ell, z)$-clustering problem, where~$\ell \leq k$. 
In order to efficiently compute each $\mathrm{Value}(B(u, r))$, the $(k, z)$-clustering algorithm of Dupré la Tour and Saulpic~\cite{latourSaulpicKz} employs the algorithm of Cohen~\cite{cohen1997size} (see~\cref{sec:greedykz} in the Appendix). 
\section{Randomized $(k, z)$-Clustering Algorithm on Graphs}
\label{sec:randomized}
In this section, we develop a $(k,z)$-clustering algorithm in the graph setting that runs in~$\tilde{O}(m)$ time, as demonstrated in~\Cref{thm:randomizedkzclustering}. \antonis{Replace forbidding loop with Dijkstra?}

\randomizedkzclustering*

Our $(k, z)$-clustering algorithm is an adaptation of the algorithm by Dupré la Tour and Saulpic~\cite{latourSaulpicKz}; a pseudocode of the $(k, z)$-clustering algorithm from~\cite{latourSaulpicKz} is provided in~\Cref{alg:simplified-greedy}. Specifically, we improve the running time by implementing some of its components more efficiently.  Consequently, we improve upon the $\tilde{O}(m^{1+\frac{1}{c}})$-time algorithm (for some constant $c \geq 5$) of Dupré la Tour and Saulpic~\cite{latourSaulpicKz}, and the $m^{1+o(1)}$-time algorithm of Jiang, Jin, Lou, and Lu~\cite{jiang2026local}.

Our adapted $(k, z)$-clustering algorithm consists of the same three components as the algorithm by Dupré la Tour and Saulpic~\cite{latourSaulpicKz} (\Cref{alg:simplified-greedy}). These three components are: (i) the ball-value estimation, (ii) the center-selection loop, and (iii) the forbidding loop. In the following paragraphs, we describe these components alongside our improvements. 

\textbf{Ball-value estimation.}
Consider a parameter $c \geq 5$, and let $\mathcal{R} \coloneqq \left\{\frac{\Delta}{(2c)^\ell} \mid \ell \in \{0, \dots, \log_{2c}(\Delta)  + 7\} \right\}$ be the set of radii considered by the algorithm.
For every vertex $u \in V$ and every radius~$r \in \mathcal{R} $, the algorithm needs access to 
$\mathrm{Value}(B(u, r))$ (see \Cref{eq:value_ball} for the definition of $\mathrm{Value}(\cdot)$). In turn, $\mathrm{Value}(\cdot)$ depends on (approximations of) the values $b(v,\tilde{r}) \coloneqq |B(v,\tilde{r})|$, where $v \in V$ is a vertex and $\tilde{r} \in \mathcal{R}$ is a radius. These values~$b(\cdot, \cdot)$ are approximated in the preprocessing phase using the algorithm by Cohen~\cite{cohen1997size}; the guarantees of Cohen's algorithm appear in~\Cref{lem:cohen-ball-count}. 

More precisely, for each pair \((u,r) \in V \times \mathcal{R}\), Cohen's algorithm returns a $(1\pm\varepsilon)$-approximate estimate~\(\tilde b(u,r)\) of \(b(u,r)\) in~\(O(\log\log n)\) time. Hence, the $(k, z)$-clustering algorithm uses the approximate size \(\tilde b(\cdot,\cdot)\) to compute~$\mathrm{Value}(\cdot)$ for all candidate balls; this component is exactly the same as in~\cite{latourSaulpicKz}.\antonis{footnote is correct, right?}\rajath{the footnote, does not consider that cohen can give $1-\epsilon$ approximate ball size}

\textbf{Center-selection loop.}
For each candidate center $u \in V$ and each radius~$r \in \mathcal{R}$, the approximate ball $N(u, 10c \cdot r)$ (in Line~\ref{line:compute_N_dijkstra} of~\Cref{alg:simplified-greedy}) is computed by running Dijkstra's algorithm from~$u$, terminating once all vertices within distance $10c \cdot r$ have been scanned. The pseudocode of this computation is provided in~\Cref{alg:center-selection}, where the ball $N(u,10c \cdot r)$ is the set of all vertices extracted from~$Q_r$, $N$ in the ~\Cref{alg:center-selection}.\footnote{Notice that in fact $B(u, 10c\cdot r)$ is computed, which is even stronger than $B(u, 10c^2 \cdot r)$ required in the definition of the approximate ball $N(u, 10c \cdot r)$ (see~\Cref{sec:greedy_kz}).} \antonis{footnote is correct, right?} Subsequently, Line~\ref{line:find_next_best_u} of~\Cref{alg:simplified-greedy} is performed by scanning the entire ball $N(u, 10c \cdot r)$.  \antonis{This is new. It's correct, right?}

Regarding the truncated Dijkstra executions (\Cref{alg:center-selection}), for each vertex $v \in V$ and each radius $r \in \mathcal{R}$, the algorithm maintains a distance label~$\mathrm{d}_r[v]$, where $\dd_r[v]$ is initialized to~$\infty$. For each radius $r \in \mathcal{R}$, the algorithm maintains a priority queue $Q_r$ that sorts the vertices based on $\dd_r[\cdot]$. These data structures are reused across all center-selection loops for radius $r$. Each execution restores both to this state before
returning: it resets $\dd_r[y] \gets \infty$ for exactly the vertices $y$
it labeled (which are the vertices of its output ball $N$, at cost
$O(|N|)$), and $Q_r$ becomes empty by construction . The shared data structures
serve to avoid initializing an array of size $n$ in every
execution.\antonis{Is this clear, or a footnode is needed?} 

\begin{algorithm}[H]
\caption{Truncated Dijkstra for ball $N(u, 10c \cdot r)$}
\label{alg:center-selection}
\KwIn{A weighted undirected graph $G=(V, E)$ with non-negative weights $w(\cdot,\cdot)$, a source $u \in V$, and a radius $r$}
\KwOut{$N = B(u, 10c \cdot r)$}
\vspace{0.2em}

$\dd_r[u] \gets 0$

Insert $u$ into $Q_r$
\vspace{0.2em}

\While{$Q_r \neq \emptyset$}{
    Extract $v$ with minimum $\dd_r[v]$ from $Q_r$ \\
    Add $v$ to $N$\;
    \vspace{0.2em}
    
    \For{each $(v,x) \in E$}{ \label{algline:relax_edge}
        $\mathit{dst} \gets \dd_r[v] + w(v,x)$
        \vspace{0.2em}
        
        \If{$\mathit{dst} < \dd_r[x] \textbf{ and } \mathit{dst} \leq 10c \cdot r$}{
            $\dd_r[x] \gets \mathit{dst}$
            
            Add $x$ to $Q_r$
        }
    }
}
\lForEach{$y \in N$}{$\dd[y] \gets \infty$ \tcp*[f]{undo all writes}}
\end{algorithm}

\textbf{Forbidding loop.} 
A single deterministic Source-Insertion SSSP algorithm is employed (from~\Cref{thm:simplified-sssp}). In particular, a source vertex~$s\notin V$ is introduced (initially with no incident edges), and the Source-Insertion SSSP algorithm is initialized on the directed version of~$G_s = (V \cup \{s\}, E, w)$.\footnote{Since~\Cref{thm:simplified-sssp} requires a directed graph as input, each undirected edge is replaced by two oppositely directed edges with the same weight.} \antonis{right?} Let~$\delta_s(\cdot)$ denote the maintained distance estimates from~$s$. 
\rajath{removed the word "fake"}
During the forbidding loop for the $i$-th center~$c_i$ (where $i \in [1, k]$), a zero-weight edge~$(s,c_i)$ is inserted into $G_s$ and is passed to the Source-Insertion SSSP algorithm. In turn, the Source-Insertion SSSP algorithm updates the distance estimates such that (where $\epsilon \in (0, 1)$): 
\[
    \dist(v,\{c_1,\ldots,c_i\}) \;\le\; \delta_s(v) \;\le\; (1 + \epsilon)\dist(v,\{c_1,\ldots,c_i\})\;\; \text{for all vertices } v\in V.
\]
Based on~\Cref{lem:inc-sssp-output}, the Source-Insertion SSSP algorithm returns a vertex set~$V_t$, consisting of the vertices whose distance estimates decreased due to the edge~$(s,c_i)$. Next, for every vertex $v \in V_t$ and every radius $r \in \mathcal{R}$, if $\delta_s(v)\le (1+\epsilon) \myhs 100 \myhs c^4\cdot r$ then the corresponding ball~$B(v, r)$ is removed (if not already removed) from the set of available balls~$\mathcal{B}$. 
\antonis{We should check the runtime of ``reported'' vertices from SSSP.}

\subsection{Analysis of the Randomized $(k, z)$-Clustering Algorithm}
In this section, our aim is to prove~\Cref{thm:randomizedkzclustering} by analyzing the $(k, z)$-clustering algorithm described in~\Cref{sec:randomized}. In particular, we analyze the running time of each component and prove that each runs in $\tilde{O}(m)$ time. The approximation ratio analysis is provided in~\Cref{sec:approxkz}, where we also show that our forbidding loop simulates that of~\Cref{alg:simplified-greedy}.\antonis{right?}

\paragraph{Running Time of the Center-Selection Loop}
In order to analyze the running time of all center-selection loops, we bound the total time of all truncated Dijkstra executions used to compute the balls $N(u, 10c \cdot r)$ over all candidate centers $u \in V$ and radii $r \in \mathcal{R}$.
We note that Lemma 3.5 in~\cite{latourSaulpicKz} (which appears as~\Cref{every-ball-once}) is used to bound the total running time of the center-selection loops in~\cite{latourSaulpicKz}.

\begin{lemma}[Lemma 3.5 in~\cite{latourSaulpicKz}] \label{every-ball-once}
    For every vertex $v \in V$ and radius $r \in \mathcal{R}$, the ball $B(v, r)$ appears at most once in the center-selection loop. 
\end{lemma}

Since we use Truncated Dijkstra's algorithm (see~\Cref{alg:center-selection}), our aim is to bound the number of times each edge is scanned (in~\Cref{lem:dijkstra-bound}). 
To achieve this, we adapt the proof of~\Cref{every-ball-once} from~\cite{latourSaulpicKz} to obtain~\Cref{lem:dijkstra_bound_vertex} and then~\Cref{lem:dijkstra-bound}. Intuitively, we argue that a ball $B(v, r)$ appears in the center-selection loop if and only if some Dijkstra execution from a candidate center $u \in V$ with radius $r \in \mathcal{R}$ (which computes the ball $N(u, 10c \cdot r)$) visits the vertex~$v$. Then based on~\Cref{every-ball-once} (our~\Cref{lem:dijkstra_bound_vertex}), for a fixed radius $r \in \mathcal{R}$, each vertex $v$ is visited by at most one truncated Dijkstra execution. In turn, each edge is scanned twice---once from each of its endpoints. The formal justification is provided in~\Cref{lem:dijkstra-bound}, and~\Cref{lem:total_dijstra_running} follows because~$|\mathcal{R}| = O(\log \Delta)$.

The next lemma is a variation of Fact 3.4 in~\cite{latourSaulpicKz}. Essentially, it states that during the $i$-th center-selection loop, the candidate center chosen for radius $r \in \mathcal{R}$ is within distance $O(r)$ of any subsequent candidate center during the $i$-th center-selection loop (assuming $c$ is a positive constant).

\begin{lemma} \label{lem:dist_ci_cand_cent}
    During the $i$-th center-selection loop, consider the candidate center $u_1 \in V$ for the radius $r_1$, and let $u_2$ be another candidate center for the radius $r_2 \leq \frac{r_1}{2c}$. Then it holds that $\dist(u_1, u_2) \leq 20c \cdot r_1$.
\end{lemma}
\begin{proof}
    The proof is the same as that of Fact 3.4 in~\cite{latourSaulpicKz}. The additional observation is that in our center-selection loops, Dijkstra's algorithm (\Cref{alg:center-selection}) computes $B(u, 10c\cdot r)$, which is even stronger than $B(u, 10c^2 \cdot r)$ required in the definition of the approximate ball $N(u, 10c \cdot r)$ (for any candidate center $u \in V$ and any radius $r \in \mathcal{R}$).\antonis{right?}
\end{proof}

\begin{lemma}\label{lem:dijkstra_bound_vertex}
    Across all center-selection loops for a fixed radius $r \in \mathcal{R}$, each vertex $v \in V$ is extracted from the priority queue $Q_r$ by at most one truncated Dijkstra execution (\Cref{alg:center-selection}). 
\end{lemma}
\begin{proof}
    For a fixed radius $r \in \mathcal{R}$, suppose to the contrary that a vertex $v \in V$
    is extracted from the priority queue $Q_r$ by two truncated Dijkstra executions (\Cref{alg:center-selection}), starting from a candidate center $u_1 \in V$ with radius $r \in \mathcal{R}$ and from a candidate center $u_2 \in V$ with radius $r \in \mathcal{R}$. In other words, the vertex $v$ is extracted from $Q_r$ during the computation of both balls $N(u_1, 10c \cdot r)$ and $N(u_2, 10c \cdot r)$.  Without loss of generality, assume that $v$ is extracted from $Q_r$ first in the execution starting from $u_1$, and later in the execution starting from $u_2$. The two executions belong to distinct center-selection loops, since within a single center-selection loop the radius strictly decreases. Hence we assume the loop of $u_1$ precedes that of $u_2$.
    
    Let $c_i$ be the center selected after the corresponding center-selection loop for $u_1$, and let~$\hat{u}_2$ and~$\hat{r}$ be the initial candidate center and radius of the corresponding center-selection loop for~$u_2$; $\hat{u}_2$ and $\hat{r}$ are defined in Line~\ref{algline:B_init_cent_sel} of~\Cref{alg:simplified-greedy}. The contradiction that we argue in the rest of the proof is that the ball $B(\hat{u}_2, \hat{r})$ could not be available at the beginning of the center-selection loop for $u_2$ (in Line~\ref{algline:B_init_cent_sel} of~\Cref{alg:simplified-greedy}).
    
    By the construction of~\Cref{alg:center-selection}, it holds that $\dist(u_1, v)\leq 10c \cdot r$.
    Based on~\Cref{lem:dist_ci_cand_cent}, we have $\dist(u_1, c_i) \leq 20c \cdot r$ (note that $c_i$ is the final candidate center of the $i$-th center-selection loop). Hence, by the triangle inequality it follows that:
    \[
        \dist(c_i, v) \;\leq\; \dist(c_i, u_1) + \dist(u_1, v) \;\leq\; 20c \cdot r + 10c \cdot r \;=\; 30c \cdot r.
    \]
    By the construction of~\Cref{alg:center-selection}, the second time that $v$ is extracted from $Q_r$ due to $u_2$, it should also hold that $\dist(u_2, v) \leq 10c \cdot r$. Thus, by the triangle inequality we infer that:
    \[
        \dist(c_i, u_2) \;\leq\; \dist(c_i, v) + \dist(v, u_2) \;\leq\; 30c \cdot r + 10c \cdot r = 40c \cdot r.
    \]
    Notice that then the ball $B(u_2, r)$ is forbidden during the forbidding loop for $c_i$, because $40c \cdot r \leq 100c^4 \cdot r$. 

    Moreover, by~\Cref{lem:dist_ci_cand_cent} we have $\dist(\hat{u}_2, u_2) \leq 20 c \cdot \hat{r}$, and thus $\dist(c_i, \hat{u}_2) \leq \dist(c_i, u_2) + \dist(u_2, \hat{u}_2) \leq 40c \cdot r + 20 c \cdot \hat{r}$. Since $r \leq \frac{\hat{r}}{2c}$, it follows that:\footnote{Note that if $r = \hat{r}$ then $\hat{u}_2 = u_2$, and the contradiction already follows since we argued that $B(u_2, r)$ has been forbidden.}
    \[
        \dist(c_i, \hat{u}_2) \;\leq\; 40c \cdot \frac{\hat{r}}{2c} + 20 c \cdot \hat{r} \;\leq\; (20 + 20c) \cdot \hat{r} \;\leq\; 40c \cdot \hat{r}.
    \]
    In turn, the ball $B(\hat{u}_2, \hat{r})$ is forbidden during the forbidding loop for $c_i$, because $40c \cdot \hat{r} \leq 100c^4 \cdot \hat{r}$. This contradicts the fact that $B(\hat{u}_2, \hat{r})$ was available at the beginning of the center-selection loop for~$u_2$. As a consequence, the vertex $v$ cannot be extracted from the priority queue $Q_r$ by two truncated Dijkstra executions.
\end{proof}

\begin{lemma}\label{lem:dijkstra-bound}
    Across all center-selection loops for a fixed radius $r \in \mathcal{R}$, each edge $e \in E$ is scanned at most twice.
\end{lemma}
\begin{proof}
Consider an arbitrary edge $e = \{v_1, v_2\} \in E$ that is scanned during the truncated Dijkstra execution (\Cref{alg:center-selection}), starting from some candidate center $u \in V$ with radius $r \in \mathcal{R}$ (which computes the ball $N(u, 10c \cdot r)$). The edge $e$ is relaxed (in Line~\ref{algline:relax_edge} of~\Cref{alg:center-selection}) only when one of its endpoints is extracted from the priority queue $Q_r$.
According to~\Cref{lem:dijkstra_bound_vertex}, for the fixed radius $r \in \mathcal{R}$, each endpoint $v_1$ and $v_2$ is extracted from $Q_r$ at most once. As a result, the edge $e = \{v_1, v_2\}$ can be relaxed at most twice in total, as required.
\end{proof}

\begin{lemma} \label{lem:total_dijstra_running}
    The total running time to compute all balls $N(u, 10c \cdot r)$ for all candidate centers $u \in V$ and radii $r \in \mathcal{R}$ in the center-selection loops is $\tilde{O}(m)$.
\end{lemma}
\begin{proof}
    Based on~\Cref{lem:dijkstra-bound}, across all center-selection loops for a fixed radius $r \in \mathcal{R}$, each edge is scanned at most twice. Since $|\mathcal{R}| = O(\log \Delta)$, each edge is scanned at most $O(\log \Delta)$ times during all center-selection loops, concluding the claim.
\end{proof}

\paragraph{Running Time of the Ball-Value Estimation and the Forbidding Loop}

Now we formally show the running time for computing $\text{Value}(u,r)$ for all vertices $u \in V$ and all radii $r \in \mathcal{R}$ using ~\Cref{lem:cohen-ball-count} and the forbidding loop using Source-Insert SSSP.
\antonis{Cohen's give expected time, we need to argue the high-probability.}
\begin{lemma} \label{lem:ball-value-runtime}
    The total running time required to compute $\text{Value}(u,r)$ for all vertices $u \in V$ and all radii $r \in \mathcal{R}$ is $O(n \cdot \log \Delta \cdot \log \log n)$,
    in addition to the preprocessing time of~\Cref{lem:cohen-ball-count}. 
\end{lemma}
\begin{proof}
    Recall that $\mathcal{R} \coloneqq \left\{\frac{\Delta}{(2c)^\ell} \mid \ell \in \{0, \dots, \log_{2c}(\Delta) + 7\} \right\}$, which means that $|\mathcal{R}| = O(\log \Delta)$. Using~\Cref{lem:cohen-ball-count}, the graph can be preprocessed to build a data structure such that for each query pair $(u, r) \in V \times \mathcal{R}$, an estimate of $|B(u,r)|$ can be returned in $O(\log \log n)$ time. Since the total number of such queries is $|V| \cdot |\mathcal{R}| = O(n \log \Delta)$, the claim follows.
\end{proof}

\begin{lemma} \label{lem:forbidding_loops_time}
    The total running time to perform all forbidding loops is $\tilde{O}(m)$.    
\end{lemma}
\begin{proof}
    A single deterministic Source-Insertion SSSP algorithm is employed, to which at most~$k \leq n$ zero-weight edges are passed.\antonis{here we need that $G$ is connected, for $n \leq m$?} Based on~\Cref{thm:simplified-sssp}, its total update time is: 
    \[
        O\left(\frac{m \log(\Delta) \log^2 n}{\varepsilon} + k\right) \;=\; \tilde{O}\left(\frac{m}{\varepsilon}\right) \;=\; \tilde{O}(m).
    \]    
    Furthermore based on~\Cref{lem:inc-sssp-output}, scanning the returned vertex set $V_t$ (consisting of vertices whose distance estimates decreased due to the zero-weight edges) takes $\tilde{O}(m)$ total time as well.
\end{proof}

\subsubsection{Approximation Ratio} \label{sec:approxkz}
Our $(k, z)$-clustering algorithm described in~\Cref{sec:randomized} is an adaptation of~\Cref{alg:simplified-greedy} by Dupré la Tour and Saulpic~\cite{latourSaulpicKz}. Since~\Cref{alg:simplified-greedy} yields an $O(1)$-approximation for the $(k,z)$-clustering problem, it remains to justify that our $(k, z)$-clustering algorithm outputs a constant-factor approximate solution as well. Recall that for every vertex $u \in V$ and every radius $r \in \mathcal{R}$, the value of the ball $B(u, r)$ satisfies:
\[
     \frac{r^z}{3} \cdot |B(u, r)| \;\leq\; \mathrm{Value}(B(u, r)) \;\leq\; 3 \myhs r^z \cdot |B(u, c \cdot r)|.
\]
Based on~\Cref{lem:cohen-ball-count}, we have $\frac{|B(u,r)|}{1+\varepsilon} \le \tilde{b}(u,r) \le (1+ \varepsilon)|B(u,r)|$. Hence, there is no additional approximation loss in the ball-value estimation procedure compared to~\cite{latourSaulpicKz}. In the center-selection loops, exact distances are used, so there is no additional approximation loss either.

\subparagraph{Forbidding loop.}
In the analysis of~\cite{latourSaulpicKz}, for each $i$-th center $c_i \in C_k$ (where $i \in [1, k]$) and every radius $r \in \mathcal{R}$, the approximate balls used in the forbidding loop of~\Cref{alg:simplified-greedy} satisfy:
\[
    B(c_i, 100c^4 \cdot r) \;\subseteq\; N(c_i, 100c^4 \cdot r) \;\subseteq\; B(c_i, 100c^5 \cdot r).
\]
In our implementation, the Source-Insertion SSSP algorithm from~\Cref{thm:simplified-sssp} maintains distance estimates $\delta_s(v)$ such that:
\[
    \dist(v,\{c_1,\ldots,c_i\}) \;\le\; \delta_s(v) \;\le\; (1 + \epsilon)\dist(v,\{c_1,\ldots,c_i\})\;\; \text{for all vertices } v\in V.
\]
Hence, if $\dist(c_i, v) \leq 100c^4 \cdot r$ then $\delta_s(v) \leq (1+\varepsilon)100c^4 \cdot r$.
Moreover, since $c \geq 5$ and~$\epsilon \in (0, 1)$, it holds that:
\[
    \dist(\{c_1, \ldots, c_i\}, v) \;\leq\; \delta_s(v) \;\leq\; (1+\varepsilon)100c^4 \cdot r \;\leq\; 100c^5 \cdot r,
\] 
where $\delta_s(v) \;\leq\; (1+\varepsilon)100c^4 \cdot r$ is due to the construction of our forbidding loop.
Thus, our implicit approximate ball $\hat{N} \coloneqq \{v \in V \mid \delta_s(v) \leq (1+\varepsilon)100c^4 \cdot r\}$ via the Source-Insertion SSSP algorithm satisfies:
\[
    B(c_i, 100c^4 \cdot r) \;\subseteq\; \hat{N} \;\subseteq\; B(\{c_1, \ldots, c_i\}, 100c^4 (1+\epsilon) \cdot r) \;\subseteq\; B(\{c_1, \ldots, c_i\}, 100c^5 \cdot r).
\]
Therefore, the use of $(1+\epsilon)$-approximate distances does not incur any additional approximation cost.

For efficiency purposes, observe that in our forbidding loop in~\Cref{sec:randomized}, only the vertex set $V_t$ is scanned; $V_t$ consists of the vertices whose distance estimates decreased due to the edge~$(s,c_i)$.
However, if some vertex $v \in \hat{N}$ does not belong to $V_t$ (i.e., $v \notin V_t$), then~$\delta_s(v)$ has been decreased earlier due to another center $c_j$ with $j \in [1, i-1]$. In turn, the corresponding ball $B(v, r)$ has already been forbidden, which concludes the correctness of our forbidding loop.

To summarize, since the analysis of~\cite{latourSaulpicKz} already tolerates constant-factor slack in these internal data structures, our modification preserves the constant-factor approximation.

\yasamin{each part of the argument should not be a seperate subsection. Add * or paragraph so they are not number. Subsection with 3 line are distracting}
\paragraph{Completion of the Proof of~\Cref{thm:randomizedkzclustering}}
Based on our preceding analysis, we can now finish the proof of~\Cref{thm:randomizedkzclustering}, which we restate for convenience.

\randomizedkzclustering*
\begin{proof}
    The time guarantee follows from~\cref{lem:cohen-ball-count,lem:ball-value-runtime,lem:total_dijstra_running,lem:forbidding_loops_time}, and the approximation and correctness guarantees follow from~\Cref{sec:approxkz}.
\end{proof}

\section{Deterministic $(k,z)$-Clustering Algorithm on Graphs}
\label{sec:deterministic} 

In this section we give a deterministic instantiation of the template
of \Cref{sec:greedy_kz}, proving \Cref{thm:deterministicclustering}. Recall
from the randomized algorithm in \Cref{sec:randomized} that the only randomized component of our
$(k,z)$-clustering algorithm is Cohen's ball-size
estimation~\cite{cohen1997size}, used to implement the ball values
; the truncated Dijkstra computations and the SSSP-based
forbidding loop are already deterministic. It therefore suffices to
estimate ball sizes deterministically. We do so using the bunches and
clusters of the Thorup--Zwick distance
oracle~\cite{thorup2005approximate}(will define bunches and
clusters in \Cref{sec:distance_oracle}). First, we give an overview of our proof, then we state Thorup--Zwick distance oracle \Cref{sec:distance_oracle} and explain how we use it to provide a deterministic algorithm. 

While distance oracles are typically used to answer distance queries
between vertex pairs with stretch $2t-1$, where $t$ is the number of
levels of the oracle (see \Cref{sec:distance_oracle}), we instead use it to
obtain, for every vertex $v \in V$ and radius $r > 0$, an approximate
neighborhood $N(v,r)$ as a union of at most $t+1$ precomputed sets (\Cref{eq:approx_ball_bunch}),
whose stored sizes yield a quantity $\hat{N}(v,r)$ satisfying
\[
  |B(v,r)| \;\le\; \hat{N}(v,r) \;\le\; (t+1) \cdot |B(v, 2\alpha r)|
  \qquad\text{with } \alpha = 4t-2
\]
(\Cref{lem:approx-ball} and \Cref{obs:approximate-ball-size-tz}). The deterministic estimation comes at a cost in
the running time as well: the oracle preprocessing takes
$\tilde{O}(t \cdot m \cdot n^{1/t})$ time, compared to the near-linear
time of the randomized algorithm. The parameter $t$ thus governs a
trade-off between solution quality and running time, and we obtain a
deterministic $O(\mathrm{poly}(t))$-approximation in
$\tilde{O}\bigl(t \cdot m \cdot n^{1/t} +
\tfrac{m}{\varepsilon}\bigr)$ time.


    \deterministicclustering*
Before describing the construction, we fix the two template parameters
for this section. Recall from \Cref{sec:greedy_kz} that the template
is governed by two parameters: $\rho \ge 1$, the slack allowed in the
ball values (\Cref{eq:value_ball}), and $c \ge 5$, the radius
blow-up allowed in the approximate balls (\Cref{eq:approx_ball}),
subject to $c \ge \rho^2/2$. The ball values that we construct in
\Cref{eq:det-value} satisfy \Cref{eq:value_ball} with
$\rho = t+1$, and the underlying
approximate neighborhoods satisfy
$B(v,r) \subseteq N(v,r) \subseteq B(v, 2\alpha r)$ with
$\alpha = 4t-2$ (\Cref{lem:approx-ball}). Accordingly, we set
\begin{equation}\label{eq:det-params}
  \rho \coloneqq t+1
  \qquad\text{and}\qquad
  c \coloneqq \max\Bigl(8t-4,\ \bigl\lceil\tfrac{(t+1)^2}{2}\bigr\rceil,
  \ 5\Bigr) = \Theta(t^2).
\end{equation}


\subsection{Thorup--Zwick distance oracle}
\label{sec:distance_oracle}
Given a weighted undirected graph $G=(V,E,w)$ and an integer
$t \ge 1$, the construction is based on a hierarchy
$V = A_0 \supseteq A_1 \supseteq \cdots \supseteq A_t = \emptyset$. For
$v \in V$ and level $i \in \{0, \dots, t-1\}$, the \emph{pivot}
$p_i(v) \coloneqq \arg\min_{w \in A_i} \dist(v,w)$ is the vertex of
$A_i$ nearest to $v$. The \emph{bunch} of $v$ is
\[
  \mathrm{Bunch}(v) \;\coloneqq\; \bigcup_{0 \le i \le t-1}
  \bigl\{w \in A_i \setminus A_{i+1} \;\big|\;
  \dist(v,w) < \dist(v, A_{i+1})\bigr\},
\]
and it contains all pivots of $v$. The \emph{cluster} of
$w \in A_i \setminus A_{i+1}$ is
$C(w) \coloneqq \{v \in V \mid \dist(v,w) < \dist(v, A_{i+1})\}$.

\begin{theorem}[Thorup--Zwick~\cite{thorup2005approximate}]
\label{thm:tz-oracle}
There is a deterministic algorithm that, given $G$ and $t \ge 1$,
computes in $\tilde{O}(t \cdot m \cdot n^{1/t})$ time the hierarchy,
all pivots, and all bunches and clusters together with the associated
exact distances, each stored as a list sorted by distance. Moreover,
$|\mathrm{Bunch}(v)| = O(t \cdot n^{1/t})$ for every $v \in V$, and the
total storage is $O(t \cdot n^{1+1/t})$.
\end{theorem}

\paragraph{Overview: Ball-size estimation via bunches and clusters.} Computing
$|B(v,r)|$ exactly for all $v \in V$ and $r \in \mathcal{R}$ is too
expensive, since a single ball may contain $\Omega(n)$ vertices.
Instead, we cover each ball by a few of the precomputed sets of
\Cref{thm:tz-oracle} and estimate its size from theirs. The covering
property we use is the following: for all $u, v \in V$, either
$u \in \mathrm{Bunch}(v)$, or there exists a level $i$ such that
$p_i(v) \in \mathrm{Bunch}(u)$ with
$\dist(v, p_i(v)) \le (4t-3) \dist(u,v)$ and
$\dist(u, p_i(v)) \le (4t-3) \dist(u,v)$ (\Cref{obs:stretch}); in
the latter case, $u \in C(p_i(v))$ by \Cref{obs:duality}. Hence,
with $\alpha := 4t-2$, every ball $B(v,r)$ is contained in the union of
$\mathrm{Bunch}(v)$ and the clusters $C(p_0(v)), \dots, C(p_{t}(v))$,
each truncated at radius $\alpha r$, and conversely this union is
contained in $B(v, 2\alpha r)$ (\Cref{lem:approx-ball}). Since each
of these $t+1$ sets is stored sorted by distance, the size of each
truncation is retrieved in $\tilde{O}(1)$ time,
and their total size $\hat{N}(v,r)$ satisfies
$|N(v,r)| \le \hat{N}(v,r) \le (t+1)\,|N(v,r)|$, since every counted
vertex lies in the union and belongs to at most one truncated cluster
per level and at most once to the bunch
(\Cref{obs:approximate-ball-size-tz}).

\begin{lemma}[Bunch--Cluster Duality~\cite{thorup2005approximate}]
\label{obs:duality}
For any $u, w \in V$: $\; w \in \mathrm{Bunch}(v) \iff v \in C(w)$.
\end{lemma}

\begin{lemma}[Distance Stretch~\cite{thorup2001compact}]
\label{obs:stretch}
For any $u, v \in V$, let $i^*$ be the smallest index  such that
$p_{i^*}(u) \in \mathrm{Bunch}(v)$. Then
\[
  \dist(u,\, p_{i^*}(u)) \;+\; \dist(p_{i^*}(u),\, v)
  \;\le\; (4t-3)\,\dist(u,v).
\]
\end{lemma}



Given $G$, apply~\Cref{thm:tz-oracle} (Thorup--Zwick--~\cite{thorup2005approximate}) 
to build the approximate distance oracle on $G$. This yields, for each vertex $v \in V$, an 
explicitly stored sorted list $\mathrm{Bunch}(v)$, and for each pivot $w$, an explicitly stored 
list $C(w)$ — both represented as doubly-linked lists, supporting $O(1)$ deletion. All cluster sizes $|C(w,\alpha r)|$ and $|\mathrm{Bunch}(v,\alpha r))|$
are precomputed and stored. The total preprocessing time is $O(mn^{1/c})$ and the total storage 
is $O(mn^{1/c})$ .


\subsection{Approximate Balls via Truncated Bunches and Clusters.}

Fix a constant $\alpha \ge 4t - 2$. For a vertex $v \in V$ and radius
$r > 0$, define
\begin{equation}
    N(v,\, r)
    \;=\;
    \mathrm{Bunch}(v,\, \alpha r)
    \;\cup\!
    \bigcup_{\substack{i \,\in\, [t] \\ d(v,\, p_i(v))\, <\, \alpha r}}
    C\!\left(p_i(v),\, \alpha r\right). \label{eq:approx_ball_bunch}
\end{equation} 
\antonis{Why do we need the bunches for $N(v, r)$? If we don't need them, it's confusing that we use them.}

where $\mathrm{Bunch}(v, \rho) = \{w \in \mathrm{Bunch}(v) : \dist(v,w) \le \rho\}$
and $C(w, \rho) = \{u \in C(w) : \dist(u,w) \le \rho\}$ denote the radius-truncated
bunch and cluster, respectively.

\paragraph{Approximate Ball Size } 
We never compute the set $N(v,r)$ explicitly: the sets in
\Cref{eq:approx_ball_bunch} may overlap, so evaluating the
cardinality of their union would require enumerating their elements,
at a cost proportional to their total size, and this for each of the
$O(n \log \Delta)$ pairs $(v,r)$ is too expensive. Instead, we only sum
the sizes of the $t+1$ truncated sets, each of which is available from
the sorted lists of \Cref{thm:tz-oracle}. The resulting quantity
counts each vertex of $N(v,r)$ once per set containing it, and hence
overcounts $|N(v,r)|$ by at most a factor of $t+1$
(\Cref{obs:approximate-ball-size-tz}):

\begin{equation}
    \hat{N}(v, r)
    =
    \left|\mathrm{Bunch}(v, \alpha r)\right|
    +
    \sum_{\substack{i \in [t] \\ d(v, p_i(v)) < \alpha r}}
    \left|C\left(p_i(v), \alpha r\right)\right|.
    \label{eq:det_ball}
\end{equation}

Our deterministic ball-value estimation is then, for every $v \in V$
and $r \in \mathcal{R}$,
\begin{equation}\label{eq:det-value}
  \mathrm{Value}(B(v,r)) \;\coloneqq\; r^z \cdot \hat{N}(v,r).
\end{equation}

By \cref{lem:approx-ball,obs:approximate-ball-size-tz} together show that
\Cref{eq:det-value} satisfies \Cref{eq:value_ball} with $\rho \coloneqq t+1$:

We will show that $\hat {N}(v,r)$ is an $t$ approximate for $|N(v,r)|$, that is $|N(v,r)| \le \hat {N}(v,r) \le t |N(v,r)|$ (see \Cref{obs:approximate-ball-size-tz}), also we show that $\hat {N}(v,r)$ can be computed in $O(t)$ time in \Cref{obs:neighborhood-runtime}, after preprocessing approximate distance oracle.



\begin{lemma}
\label{lem:approx-ball}
For every $v \in V$ and $r > 0$,
$B(v,\, r) \;\subseteq\; N(v,\, r) \;\subseteq\; B(v,\, 2\alpha r)$, where $\alpha \geq 4t - 2$ and $t$ is the number of levels in the Thorup--Zwick hierarchy.
\end{lemma}

\begin{proof}

We prove the two containments separately.

\textbf{Part 1: First containment: $B(v,r) \subseteq N(v,r)$.}
Let $u \in B(v,r)$, so $\dist(u,v) \le r$. We consider two cases:

\emph{Case 1: $u \in \mathrm{Bunch}(v)$.}
Since $\dist(u, v) \leq r \leq \alpha r$, we have $u \in \mathrm{Bunch}(v, \alpha r) \subseteq N(v, r)$.

\emph{Case 2: $u \notin \mathrm{Bunch}(v)$.} We will show that $u$ appears in $C(p_i(v), \alpha r)$ for some $i$ with $d(v, p_i(v)) < \alpha r$.
By ~\Cref{obs:stretch}, let $i^*$ be the largest index such that $p_{i^*}(v) \in \mathrm{Bunch}(u)$, and set $w := p_{i^*}(v)$. Then:
\begin{equation}
  \label{eq:stretch-bound}
  \dist(v, w) \;\le\; (4t-3)\,\dist(u,v) \;\le\; \alpha\,r.
\end{equation}

Hence $d(v,p_{i^*}) < \alpha r$.
Now, we will establish $u \in C(w, \alpha r)$:  By ~\Cref{thm:tz-oracle}, $w = p_{i^*}(v) \in \mathrm{Bunch}(u)$.
By ~\Cref{obs:duality} (applied with the roles of $u$ and $w$ exchanged),
$w \in \mathrm{Bunch}(u) \Rightarrow u \in C(w)$. Since $\dist(u,w) \le (4t-3)r \le \alpha r$ by Observation~\eqref{eq:stretch-bound}
, we get $u \in C(w,\alpha r)$. Hence $u \in N(v,r)$.


\bigskip
\textbf{Part 2: Second containment: $N(v,r) \subseteq B(v, 2\alpha r)$.}
Let $u \in N(v,r)$. We consider two cases:

\emph{Case 1: $u \in \mathrm{Bunch}(v, \alpha r)$:}
By definition of $\mathrm{Bunch}(v, \alpha r)$ , $\dist(v,u) \le \alpha r \le 2\alpha r$.

\emph{Case 2: $u \in C(p_i(v), \alpha r)$ for some $i \in [t]$ with $d(v, p_i(v)) < \alpha r$:}
We have $d(v, p_i(v)) < \alpha r$ (by the condition in the definition of $N(v,r)$) and
$d(u, p_i(v)) \le \alpha r$ (since $u \in C(p_i(v), \alpha r)$). By the triangle inequality,
$
  d(u,v) \le d(u, p_i(v)) + d(p_i(v), v) \le \alpha r + \alpha r = 2\alpha r.
$
Hence $u \in B(v, 2\alpha r)$.
\end{proof}

\begin{lemma}

 \label{obs:approximate-ball-size-tz} For a fixed $v$ and $r$, observe that
  $ |N(v,r)| \le 
  \hat{N}(v,r)
  \;\le\;
  (t+1)\cdot|N(v,\alpha r)|
$ .That is, the sum of precomputed cluster sizes is a $(t+1)$-approximation of
$|N(v,r)|$; 
\end{lemma}

\begin{proof}
The first inequality follows from the union bound. For the second inequality, note that each
vertex $u \in N(v,r)$ appears in at most $t$ of the clusters $C(p_i(v), \alpha r), i \in [t]$ (one for
each level $i$) and at most once in the $ \mathrm{Bunch}(v,r)$. 
\end{proof}

\begin{lemma}

\label{lem:total-value-time}
For all $v \in V$ and all radii $r \in \bigl\{\Delta/(2c)^i :
i = 0,\ldots,\log_{2c}\Delta + 7\bigr\}$, the values
$\{\hat{N}(v,r)\}$ can be computed in total time
$O(t \cdot n \log \Delta)$.
\end{lemma}
\begin{proof}

There are $n$ vertices and $O(\log \Delta)$ candidate radii, giving
$O(n \log \Delta)$ queries in total. We will prove below for each $(v,r)$ takes $O(t)$ time \Cref{obs:neighborhood-runtime}, Hence the total time of taken is $O(t \cdot n \log \Delta)$.

\begin{observation}
\label{obs:neighborhood-runtime}
For any vertex $v \in V$ and radius $r \ge 0$, the quantity $\hat{N}(v,r)$
can be computed in $O(t)$ time.
\end{observation}

\begin{proof}
The value $|\mathrm{Bunch}(v,\alpha r)|$ can be accessed in constant time.
The summation involves at most $t$ terms, and each cluster size
$|C(p_i(v),\alpha r)|$ is precomputed and can be retrieved in $O(1)$ time.
Thus, the total time to evaluate the expression is $O(t)$.
\end{proof}




\end{proof}

\bibliographystyle{plainurl}
\bibliography{references}

\newpage
\appendix
\section*{Appendix}
\section{Approximation of $\mathrm{Value}(\cdot, \cdot)$} \label{sec:approx-value}
The static $(k, z)$-clustering algorithm of Dupre la Tour and Saulpic~\cite{latourSaulpicKz} defines
a function $\mathrm{Value}(B(u, r))$ which satisfies the following inequalities:
\[
    \frac{r^z}{3} \cdot |B(u, r)| \;\leq\; \mathrm{Value}(B(u, r)) \;\leq\; 3 \myhs r^z \cdot |B(u, c \cdot r)|,
\]
for a vertex $u \in V$, a positive real number $r$, a constant $z \geq 1$, and a constant $c \geq 5$.
The arguments presented in Appendix~A of~\cite{latourSaulpicKz} can be generalized to work with any constant in the approximation ratio, not just $3$. Specifically, it suffices that~$\mathrm{Value}(B(u, r))$ satisfy the following inequalities:
\[
    \frac{r^z}{\rho} \cdot |B(u, r)| \;\leq\; \mathrm{Value}(B(u, r)) \;\leq\; \rho \myhs r^z \cdot |B(u, c \cdot r)|,
\]
for any parameter $\rho \geq 1$ and $c \geq \frac{\rho^2}{2}$. The approximation ratio of the algorithm is then affected by a factor of $\rho^2$; hence for constant parameter $\rho$, the approximation ratio remains a constant factor.
The same arguments in Appendix A of~\cite{latourSaulpicKz} can be repeated with $\rho$ instead of $3$; for completeness we indicate the lemmas where the necessary adjustments occur:
\begin{enumerate}
    \item The first lemma that should be adjusted is Lemma A.1 in~\cite{latourSaulpicKz}, where within its proof, the constant $3$ can be trivially replaced by the parameter $\rho$. In turn, the constraint on $c$ becomes $c \geq \frac{\rho^2}{2}$ instead of $c \geq 5$.

    \item The second lemma that should be adjusted is Lemma A.5 in~\cite{latourSaulpicKz}. When $\cost(\text{In}(P_\gamma), C_k)$ is analyzed in the proof of Lemma A.5, the constant $3$ can be trivially replaced by $\rho$, since $(r_\gamma)^z \cdot |B(\gamma, r_\gamma)| \;\leq\; \rho \cdot \mathrm{Value}(B(\gamma, r_\gamma))$. \antonis{They write $p_\gamma$, is it a typo?}

    \item The third lemma that should be adjusted is Lemma A.6 in~\cite{latourSaulpicKz}. When $\cost(B(x, c \cdot r), \Gamma)$ is analyzed in the proof of Lemma A.6, the constant $3$ can be trivially replaced by $\rho$, since $r^z \cdot |B(x, c \cdot r)| \;\geq\; \frac{\mathrm{Value}(B(x, r))}{\rho}$.

    \item Finally, by combining Lemmas A.5 and A.6 at the end of Appendix~A in~\cite{latourSaulpicKz}, the extra $\rho^2$ factor arises in the approximation ratio.\antonis{We have only $\rho$ if Value() does not underestimate.}
\end{enumerate}

\section{\boldmath$k$-Center with Outliers}
\label{sec:outliers}

We now extend the incremental framework developed in ~\Cref{sec:kcenter} to the $k$-center problem with outliers, where up to $t$ points may be
discarded. As before, the key primitive is the incremental SSSP structure of ~\Cref{thm:simplified-sssp}. We give two results (i) ~\Cref{lem:outlier-2approx} handles the case of constant $k$ and achieves an $2 + \varepsilon$-approximation, $\varepsilon > 0$, while having $(1+\epsilon)t$ outliers (ii)~\Cref{lem:outlier-bicriteria} handles general
$k$ via a randomized bi-criteria guaranty . Both run in $\tilde{O}(\frac{m}{\varepsilon})$ time.

\subsection{2-Approximation for Constant \boldmath$k$}
\label{subsec:outlier-2approx}

\kcenteroutlier*

\begin{proof}
We adapt the $2$-approximation algorithm of Ding, Yu, and
Wang~\cite{ding2019greedy} for $k$-center with $(1+\epsilon)t$ outliers. The algorithm proceeds in $k$ rounds.
In each round it identifies the $\lceil(1+\varepsilon)t\rceil$ vertices
currently farthest from the center set $C$ --- these are the candidate outliers
--- and selects one of them uniformly at random as the next center.
Maintaining the farthest vertices requires knowing $\mathrm{dist}(v, C)$ for
all $v \in V$ after each center addition, which we handle incrementally: upon
adding center $c_j$ to $C$, we insert the zero-weight edge $(s, c_j)$ into the
SSSP structure of ~\Cref{thm:simplified-sssp}.
The full procedure is stated as ~\Cref{alg:outlier-2approx}. Let $\hat{d}(\cdot)$ denote the current distance estimates from $s$. All vertices are stored in a max-heap keyed by $\hat d(v)$.

\begin{algorithm}[h]
\caption{$2 + \varepsilon $-Approximation for $k$-Center with $t$ Outliers}
\label{alg:outlier-2approx}
\KwIn{Weighted graph $G = (V,E,w)$; integers $k$, $t$; parameter
      $\varepsilon > 0$}
\KwOut{A set $C \subseteq V$ of $k$ centers}

Construct $ G' = (V \cup \{ s \}, E, w )$\;

Initialize the $(1+\frac{\varepsilon}{2})$-approximate incremental SSSP structure from $s$ (~\Cref{thm:simplified-sssp})\;
$C \leftarrow \varnothing$\;
Sample $c_1 \in V$ uniformly at random, set $C \leftarrow \{c_1\}$, 
and insert edge $(s,c_1)$ with weight $0$\;
\For{$j = 2$ \KwTo $k$}{
    
    Let $Q_j \subseteq V$ be the $\lceil(1+\varepsilon)t\rceil$ vertices with
    the largest values of $\hat{d}(\cdot)$\;
    Uniformly at random select one vertex $c_j$ from $Q_j$\;
    $C \leftarrow C \cup \{c_j\}$\;
    Insert edge $(s, c_j)$ with weight $0$.
}
\Return $C$\;
\end{algorithm}

\begin{lemma}
\label{lem:outlier-2approx-approx}
Let $G=(V,E,w)$ be a weighted graph, and let $C$ be the set of centers
returned by ~\Cref{alg:outlier-2approx}.
Then, with probability at least probability at least $(1-\gamma)(1/(1+\varepsilon)t)^{k-1}$, $C$ induces an $(2+\varepsilon)$-approximate
solution to the $k$-center problem with $t$ outliers.
\end{lemma}

\begin{proof}

 Let $C^* = \{o_1, \dots, o_k\}$ be an optimal solution, inducing clusters $O_1, \dots, O_k$, and let $T^* \subseteq V$ denote the set of $t$ optimal outliers. Let $\mathrm{OPT}$ be the optimal radius, i.e.,
$
\forall i \in [k],\ \forall v \in O_i:\quad \mathrm{dist}(v, o_i) \le \mathrm{OPT}.
$

Let $C_j = \{c_1, \dots, c_j\}$ be the set of centers chosen after $j$ iterations. At iteration $j+1$, let
$
Q_{j+1} := \operatorname{Top}_{\lceil (1+\varepsilon)t \rceil}
\bigl(\{ \hat{d}(v, C_j) : v \in V \}\bigr)
$
be the set of the $\lceil (1+\varepsilon)t \rceil$ farthest points from $C_j$, as stored by $\hat{d}(\cdot)$.

\paragraph{Inliers and outliers.}
Let $T^* \subseteq V$ denote the set of $t$ outliers in the optimal solution. The remaining points 
$
V \setminus T^*
$
are called \emph{inliers}. Equivalently, the set of inliers is precisely the union of the optimal clusters:
$
\bigcup_{i=1}^k O_i = V \setminus T^*$. Let
$\lambda_j$ denote the number of covered clusters after round $j$.

\begin{lemma}[Cluster diameter bound]
\label{lem:cluster-diameter}
If $\{u,v\} \in O_i$, then 
$\mathrm{dist}(u, v) \le 2\mathrm{OPT}. $
\end{lemma}
\begin{proof}
For any $\{u,v\} \in O_i$, we have $\mathrm{dist}(v, o_i) \le \mathrm{OPT}$ and $\mathrm{dist}(u, o_i) \le \mathrm{OPT}$. The claim follows by triangle inequality.
\end{proof}

\paragraph{Covered clusters.}
We say that a cluster $O_i$ is \emph{covered} at step $j$ if $C_j \cap O_i \neq \emptyset$, and \emph{uncovered} otherwise.

\begin{lemma}
\label{lem:uncovered-in-Q}
At step $j \ge 1$, either (1) every vertex of $Q_{j+1}$ is at distance estimate $\hat{d}(\cdot)$  greater than
  $(2 + \epsilon)\mathrm{OPT}$ from $C_j$;  in this case $Q_{j+1}$ contains at least
  $\lceil(1+\varepsilon)t\rceil-t\ \ge\ \lceil\varepsilon t\rceil\ \ge 1$
  inliers, all of which belong to uncovered clusters $(2) \min_{u\in Q_{j+1}}\hat{d}_j(u)\le(2+\varepsilon)\mathrm{OPT}$.  

\end{lemma}

\begin{proof}
Suppose $(2)$ is false, i.e.\ $\hat{d}_j(u)>(2+\varepsilon)\mathrm{OPT}$ for all $u\in Q_{j+1}$. If such a $u$ were an inlier of a
\emph{covered} cluster $O_i$, then picking $c\in C_j\cap O_i$ and using
Lemma~\ref{lem:cluster-diameter} would give $\hat{d}_j(u) \le (2 + \epsilon) \mathrm{OPT}$, a contradiction. Hence $Q_{j+1}$ contains
no inlier of a covered cluster, so
$Q_{j+1}\subseteq T^*\cup\{\text{inliers of uncovered clusters}\}$, and
the number of uncovered-cluster inliers in $Q_{j+1}$ is at least
$|Q_{j+1}|-|T^*|=\lceil(1+\varepsilon)t\rceil-t\ge\lceil\varepsilon
t\rceil\ge 1$ (using $t\ge1$ and that $t$ is an integer).

Suppose $(1)$ is false, $(2)$ is trivially true.




\end{proof}

\begin{lemma} \label{lem:already_2_aprox}
    If condition $(2)$ holds of \Cref{lem:uncovered-in-Q}, then in
  this case at most $\lceil(1+\varepsilon)t\rceil-1$ vertices $v\in V$
  satisfy $\hat{d}_j(v)>(2+\varepsilon)\mathrm{OPT}$, and — since estimates are
  non-increasing — the same bound holds at every step $j'\ge j$.
\end{lemma}
\begin{proof}
    Since vertices with estimate exceeding
$(2+\varepsilon)\mathrm{OPT}$ lie in $Q_{j+1}\setminus\{u\}$, and there are at
most $\lceil(1+\varepsilon)t\rceil-1$ of them. Since each $\hat{d}(v)$ is
non-increasing as centers are inserted, the set of vertices exceeding
the threshold can only shrink, so the bound persists for all
$j'\ge j$.
\end{proof}
\begin{lemma}[Progress per iteration]
\label{lem:progress}
Conditioned on the event that $(1)$ of \Cref{lem:uncovered-in-Q} holds for $j > 1$, the probability that iteration $j+1$ selects a vertex from an uncovered cluster is at least $\frac{\varepsilon}{\lceil (1+\varepsilon) \rceil}.$
\end{lemma}

\begin{proof}
By $(1)$ of Lemma~\ref{lem:uncovered-in-Q}, at least
$\lceil(1+\varepsilon)t\rceil-t$ of the
$\lceil(1+\varepsilon)t\rceil$ vertices of $Q_{j+1}$ are inliers of
uncovered clusters.  Since $c_{j+1}$ is chosen uniformly at random from $Q_{j+1}$, the claim follows.
\end{proof}

\begin{lemma}
\label{lem:covering}
With probability at least $(1-\gamma)(\frac{\epsilon}{\lceil (1+\varepsilon)\rceil})^{k-1}$,
the output $C_k$ of \Cref{alg:outlier-2approx} is a $(2+\varepsilon)$-approximation

\end{lemma}

\begin{proof}
We will prove this by showing that after $k$ iterations, all clusters $O_1, \dots, O_k$ are covered with probability at least $(1 - \gamma) \cdot ( \frac{\epsilon}{\lceil (1+\varepsilon) \rceil} )^{k-1}$,
Since $c_1$ is sampled uniformly from $V$, $\Pr[c_1 \notin T^*] = 1 -\gamma$; condition on this event, so $\lambda_1 = 1$. At each iteration $1<j \le k$, if event  $(1)$ of \Cref{lem:uncovered-in-Q} holds, then by ~\Cref{lem:progress}, the probability of covering a new cluster is at least $p := \frac{\epsilon}{\lceil (1+\varepsilon) \rceil}$, else by \Cref{lem:already_2_aprox}, we already have a $(2 + \varepsilon)$ approximate with having $(1+\epsilon)t$ outliers. As, we run it for at most $k-1$ iterations, we get the success probability as $( \frac{1}{\lceil (1+\varepsilon)t \rceil} )^{k-1}$.  Hence for the \Cref{alg:outlier-2approx}, the success probability is at least $(1 - \gamma) \cdot ( \frac{\epsilon}{\lceil (1+\varepsilon) \rceil} )^{k-1}$. As, every $O_i$ is covered, by \Cref{lem:cluster-diameter}, we get $(2+\epsilon)\mathrm{OPT}$ approximation. We have the $(1+\epsilon)t$ vertices with the largest distance estimate from the centers $C_k$ as outliers. 

\end{proof}


\begin{lemma}
\label{lem:outlier-2approx-runtime}
~\Cref{alg:outlier-2approx} runs in $\tilde{O}(\frac{m}{\varepsilon})$ total time.
\end{lemma}
\begin{proof}

Exactly $k$ zero-weight edges are inserted into the SSSP structure, one per
center.
Since $k \le m$ , the total number of SSSP updates is at most $\Tilde{O}(\frac{m}{\varepsilon})$, and by
~\Cref{thm:simplified-sssp} guarantees that the total update time is $\Tilde{O}(\frac{m}{\varepsilon})$. By \Cref{lem:inc-sssp-output}, the number of vertices whose distance
estimate ever changes is $\Tilde{O}(\frac{m}{\varepsilon})$, so the total cost of all
decrease-key operations on the Heap is $\Tilde{O}(\frac{m}{\epsilon})$ time.
Identifying $Q_j$ at each step uses $\lceil(1+\varepsilon)t\rceil$
extract-max and re-insert operations on heap structure, costing $O(t\log n)$ per
iteration.
Hence the overall running time is $\tilde{O}(\frac{m}{\varepsilon})$.
\end{proof}

\end{proof}
\end{proof}

\subsection{Randomized Bicriteria Approximation for General \boldmath$k$}
\label{subsec:outlier-bicriteria}

We give a bicriteria algorithm with following properties.

\begin{lemma}
\label{lem:outlier-bicriteria}
    There exists a randomized algorithm that computes a
    $(2 + \varepsilon ,\, O(1/\varepsilon ))$-bicriteria approximate solution to the $k$-center
    problem with $t$ outliers on $G = (V, E, w)$ in $\tilde{O}(\frac{m}{\varepsilon})$ time, with
    high probability.
    Specifically, the returned set $C$ has covering radius at most
    $(2+ \varepsilon ) \cdot \mathrm{OPT}$ while having at most $O(\frac{k \log \frac{1}{\eta}}{\varepsilon} )$ points
    as centers.
\end{lemma}

\begin{proof}

The proof follows similar to \Cref{lem:outlier-2approx}, but instead of \Cref{lem:progress}, we have the following \Cref{obs:k-center-outlier-bic}. Let $\lambda_j$ count how many optimal clusters intersect the current center set $C$ at iteration $j$. Initially, $\lambda_1 \geq 1$ with probability $1-\eta$ by random sampling.

\begin{observation}\label{obs:k-center-outlier-bic}
    In each iteration $j$, when random sampling from $Q_j$, with probability at least $1-\eta$, the algorithm selects a vertex from uncovered cluster, increasing $\lambda_j$.  
\end{observation}
\begin{proof}
Fix an iteration $j$ in which at least one optimal cluster is uncovered.
We want to show that with probability at least $1 - \eta$, at least one
of the sampled vertices in this iteration belongs to an uncovered cluster.

\paragraph*{Step 1: Uncovered cluster inliers dominate $Q_j$.}
Recall that $Q_j$ consists of the $\lceil(1+\varepsilon)t\rceil$ vertices
with the largest distance estimates $\hat{d}(\cdot)$ from the current
center set. We argue that the points in $Q_j$ that do \emph{not} belong
to any uncovered cluster number at most $t$.

Indeed, any inlier of a \emph{covered} cluster is within $2 \cdot
\mathrm{OPT}$ of some current center (by the triangle inequality through
its optimal center), so its distance estimate is small and it does not
appear in $Q_j$. The only remaining candidates for $Q_j$ that are not
from uncovered clusters are the $t$ outliers. Hence at most $t$ points
in $Q_j$ are non-uncovered, and at least
$\lceil(1+\varepsilon)t\rceil - t \geq \varepsilon t$
points in $Q_j$ are inliers of uncovered clusters.

\paragraph*{Step 2: Probability that a single sample misses uncovered clusters.}
Since at most $t$ out of $\lceil(1+\varepsilon)t\rceil$ points in $Q_j$
are non-uncovered, the probability that a single uniform sample from
$Q_j$ does \emph{not} land in an uncovered cluster is at most
$\frac{t}{\lceil(1+\varepsilon)t\rceil} \leq \frac{1}{1+\varepsilon}$.

\paragraph*{Step 3: Probability that all samples miss.}
In iteration $j$, the algorithm draws
$q := \left\lceil\frac{1+\varepsilon}{\varepsilon}
\ln\!\left(\frac{1}{\eta}\right)\right\rceil$
vertices independently and uniformly from $Q_j$. The probability that
\emph{all} of them miss every uncovered cluster is at most
$\left(\frac{1}{1+\varepsilon}\right)^{q}
\leq \left(\frac{1}{1+\varepsilon}\right)^{
\frac{1+\varepsilon}{\varepsilon}\ln(1/\eta)}$.
$\Pr[\text{all sampled vertices miss uncovered clusters}]
\leq e^{-\ln(1/\eta)} = \eta$.

Therefore, with probability at least $1 - \eta$, at least one sampled
vertex in iteration $j$ belongs to an uncovered cluster, which increases
$\lambda_j$ by at least one.
\end{proof}

\paragraph{Covering all clusters.}
Setting $\eta = \delta/k$ for a desired failure probability $\delta \in 
(0,1)$, \Cref{obs:k-center-outlier-bic} guarantees that in each 
iteration $\lambda_j$ increases by at least one with probability $1 - 
\delta/k$. By a union bound over all $k$ iterations, all optimal 
clusters are covered simultaneously with probability at least 
$1 - k \cdot (\delta/k) = 1 - \delta$. Setting $\delta = 1/n$ gives 
a high probability guarantee, with the number of samples per iteration 
increasing only by a $\log(kn)$ factor, which is absorbed into the 
$\tilde{O}$ notation.

The full procedure is stated as ~\Cref{alg:outlier-bicriteria}.  Let $\hat{d}(\cdot)$ denote the current distance estimates from $s$.

\begin{algorithm}[H]
\caption{Randomized Bicriteria Approximation for $k$-Center with $t$ Outliers}
\label{alg:outlier-bicriteria}
\KwIn{Weighted graph $G=(V,E,w)$; integers $k$, $t$;
      parameters $\varepsilon > 0$, $\eta \in (0, \tfrac{1}{2})$}
\KwOut{A set $C \subseteq V$}
Add a super-source $s$; initialize the $(1{+}\varepsilon)$-approximate
incremental SSSP structure from $s$ (~\Cref{thm:simplified-sssp})\;
$\gamma \leftarrow t/n$\;
$C \leftarrow \varnothing$\;
Uniformly at random select $\left\lceil\frac{1}{1-\gamma}
\ln\!\left(\frac{1}{\eta}\right)\right\rceil$ vertices from $V$, add them to
$C$, and insert the corresponding zero-weight edges into the SSSP structure\;
\For{$j = 1$ \KwTo $k$}{
   
    Let $Q_j \subseteq V$ be the $\lceil(1+\varepsilon)t\rceil$ vertices with
    the largest values of $\hat{d}(\cdot)$\;
    Uniformly at random select
    $\left\lceil\frac{1+\varepsilon}{\varepsilon}
    \ln\!\left(\frac{1}{\eta}\right)\right\rceil$ vertices from $Q_j$, add them
    to $C$, and insert the corresponding zero-weight edges into the SSSP
    structure\;
}
\Return $C$\;
\end{algorithm}

\paragraph{Running time.}

\textbf{Maintaining the max‑heap of distance estimates.}
        All vertices are stored in a max‑heap keyed by their current
        estimate $\widehat d(v)$.  Whenever the incremental SSSP structure
        relaxes an edge $(u,v)$ and consequently lowers $\widehat d(v)$,
        we perform a decrease‑key operation on $v$ (cost $O(\log n)$).
        The total number of such heap updates is bounded by $\Tilde{O}(m)$ ~\Cref{lem:inc-sssp-output}.
        
        \textbf{Extracting the $\lceil(1+\varepsilon)t\rceil$ farthest vertices.}
        At the beginning of each of the $k$ iterations we need the set
        $Q_i$ of the $\lceil(1+\varepsilon)t\rceil$ vertices with the
        largest current estimate.
        Using the max‑heap we obtain $Q_i$ by performing
        $\lceil(1+\varepsilon)t\rceil$ extract‑max operations and
        immediately reinserting the extracted vertices.
        Each extraction (or reinsertion) costs $O(\log n)$,
        hence a single iteration costs $O(t\log n)$.

Hence the overall running time is $\tilde{O}(m)$.
\end{proof}
\section{Cohen's Algorithm} \label{sec:greedykz}

\begin{lemma}[Cohen's ball-size estimation, Theorem 5.2 in~\cite{cohen1997size}]
\label{lem:cohen-ball-count}
There is a randomized algorithm that, given a weighted undirected graph
$G = (V,E,w)$ and a small constant $\varepsilon > 0$, constructs a data structure with the
following properties:

\begin{itemize}
    \item The expected preprocessing time is $O\!\left(\frac{m \log^{2} n \,+\, n \log^{3} n}{\varepsilon^2}\right)$.

    \item For any query pair $(v,r) \in V \times \mathbb{R}^{+}$, the data structure returns an
    estimate $\tilde{b}(v,r)$ of~$|B(v,r)|$, where $B(v,r)$ denotes the ball of radius $r$ around $v$. 

    \item The expected query time is $O(\log \log n)$.

    \item With high probability, for all $(v,r) \in V \times \mathbb{R}^{+}$ it holds that:
    \begin{center}
        $\dfrac{\left|\,|B(v,r)| - \tilde{b}(v,r)\,\right|}{|B(v,r)|} \;\le\; \varepsilon$.
    \end{center}
\end{itemize}
\end{lemma}

\end{document}